%% file: 00_main.tex
\documentclass[conference]{IEEEtran}

\IEEEoverridecommandlockouts
\usepackage{tabularx}
\usepackage{caption}
\usepackage{multirow}
\usepackage{placeins}
\usepackage{float}
\usepackage{amsmath}
\usepackage{kotex}
\usepackage{algorithmic}
\usepackage{graphicx}
\usepackage{subcaption}
\usepackage{lipsum}
\usepackage{array}
\usepackage{textcomp}
\usepackage{stfloats}
\usepackage{makecell} 
\usepackage{url}
\usepackage{verbatim}
\usepackage{xcolor}
\usepackage{amsthm}
\usepackage[normalem]{ulem}
\usepackage{kotex}
\usepackage{diagbox} 
\usepackage{booktabs}
\usepackage{multirow} 
\usepackage{mathtools} 
\usepackage[linesnumbered,ruled,vlined]{algorithm2e}
\usepackage{hyperref}
\usepackage{orcidlink} 
\newcommand{\spara}[1]{\noindent{\bf #1}}
\DeclareRobustCommand{\hourglass}
{\mathrel{\mathpalette\hour@glass\relax}}
\hypersetup{
    colorlinks=false,
    hidelinks
}
\usepackage{enumitem}
\newcommand{\flexi}{\textsf{Flexi$-$clique}}
\newcommand{\NPA}{\textsf{NPA}}

\newcommand{\FPA}{\textsf{FPA}}

\newcommand{\EBA}{\textsf{EBA}}
\newcommand{\FastQC}{\textsf{FastQC}}

\newcommand{\CP}{$\mathrm{CP}$-Branching}

\newtheorem{example}{Example}
\newtheorem{theorem}{Theorem}

\newtheorem{problemDefinition}{Problem}
\newtheorem{definition}{Definition}
\newtheorem{property}{Property}
\newtheorem{rules}{\textbf{Rule}}

\newtheorem{lemma}{Lemma}
\newtheorem{remark}{Remark}



\def\BibTeX{{\rm B\kern-.05em{\sc i\kern-.025em b}\kern-.08em
    T\kern-.1667em\lower.7ex\hbox{E}\kern-.125emX}}
\begin{document}

\title{Efficient Computation of Maximum Flexi-Clique in Networks}

\author{
Song Kim\textsuperscript{1},
Hyewon Kim\textsuperscript{1},
Kaiqiang Yu\textsuperscript{2},
Taejoon Han\textsuperscript{1},
Junghoon Kim\textsuperscript{1},
Susik Yoon\textsuperscript{3},
Jungeun Kim\textsuperscript{4}
\\[0.8em]

\textsuperscript{1}\textit{Department of Computer Science and Engineering, UNIST, Ulsan, South Korea}\\
\textsuperscript{2}\textit{State Key Laboratory of Novel Software Technology, Nanjing University, Nanjing, China}\\
\textsuperscript{3}\textit{Department of Computer Science and Engineering, Korea University, Seoul, South Korea}\\
\textsuperscript{4}\textit{Department of Computer Engineering, Inha University, Incheon, South Korea}
}

\maketitle

\begin{abstract}
Discovering large cohesive subgraphs is a key task for graph mining. Existing models, such as clique, $k$-plex, and $\gamma$-quasi-clique, use fixed density thresholds that overlook the natural decay of connectivity as the subgraph size increases. The {\flexi} model overcomes this limitation by imposing a degree constraint that grows sub-linearly with subgraph size. We provide the algorithmic study of {\flexi}, proving its NP-hardness and analysing its non-hereditary properties. To address its computational challenge, we propose the Flexi-Prune Algorithm (\FPA), a fast heuristic using core-based seeding and connectivity-aware pruning, and the Efficient Branch-and-Bound Algorithm (\EBA), an exact framework enhanced with multiple pruning rules. Experiments on large real-world and synthetic networks demonstrate that {\FPA} achieves near-optimal quality at much lower cost, while {\EBA} efficiently computes exact solutions. {\flexi} thus provides a practical and scalable model for discovering large, meaningful subgraphs in complex networks.
\end{abstract}


\input{01_introduction}
\input{02_preliminaries}

\input{03_problem_statement}

\input{04_algorithm_heuristics}

\input{05_algorithm_BB}

\input{06_experiments}

\input{08_conclusion}

\bibliographystyle{IEEEtran}
\bibliography{name}

\end{document}

%% file: 01_introduction.tex
\section{INTRODUCTION}\label{sec:introduction}

Understanding what makes a subgraph significant is a fundamental problem in cohesive subgraph discovery.
Among various structural factors, subgraph size has been repeatedly shown to play a decisive role~\cite{wuchty2007increasing}.
Empirical studies in social network analysis indicate that larger groups tend to exert greater collective influence.
For instance, Metcalfe's law~\cite{metcalfe2013metcalfe} suggests that the value of a network grows super-linearly with its size, while Tverskoi et al.~\cite{tverskoi2021dynamics} report that the impact of a community increases with the number of its members.
These findings have naturally led to extensive research on identifying maximum cohesive subgraphs, giving rise to classical models such as clique~\cite{carraghan1990exact}, $k$-plex~\cite{balasundaram2011clique}, and quasi-clique~\cite{pattillo2013maximum}.

Despite their success, existing cohesive subgraph models rely on fixed connectivity constraints that are independent of subgraph size. The clique model enforces complete adjacency, which severely limits scalability in large networks~\cite{leskovec2020mining,chang2019efficient}. Relaxations such as $k$-plex and $\gamma$-quasi-clique loosen this requirement by allowing missing edges or enforcing a fixed density threshold. However, all these models impose size-invariant constraints, making them unable to reflect the systematic trade-off between subgraph size and internal connectivity observed in real-world networks.

\begin{figure}[t]
\centering
    \begin{subfigure}{.47\linewidth}
    \includegraphics[width=0.99\linewidth]{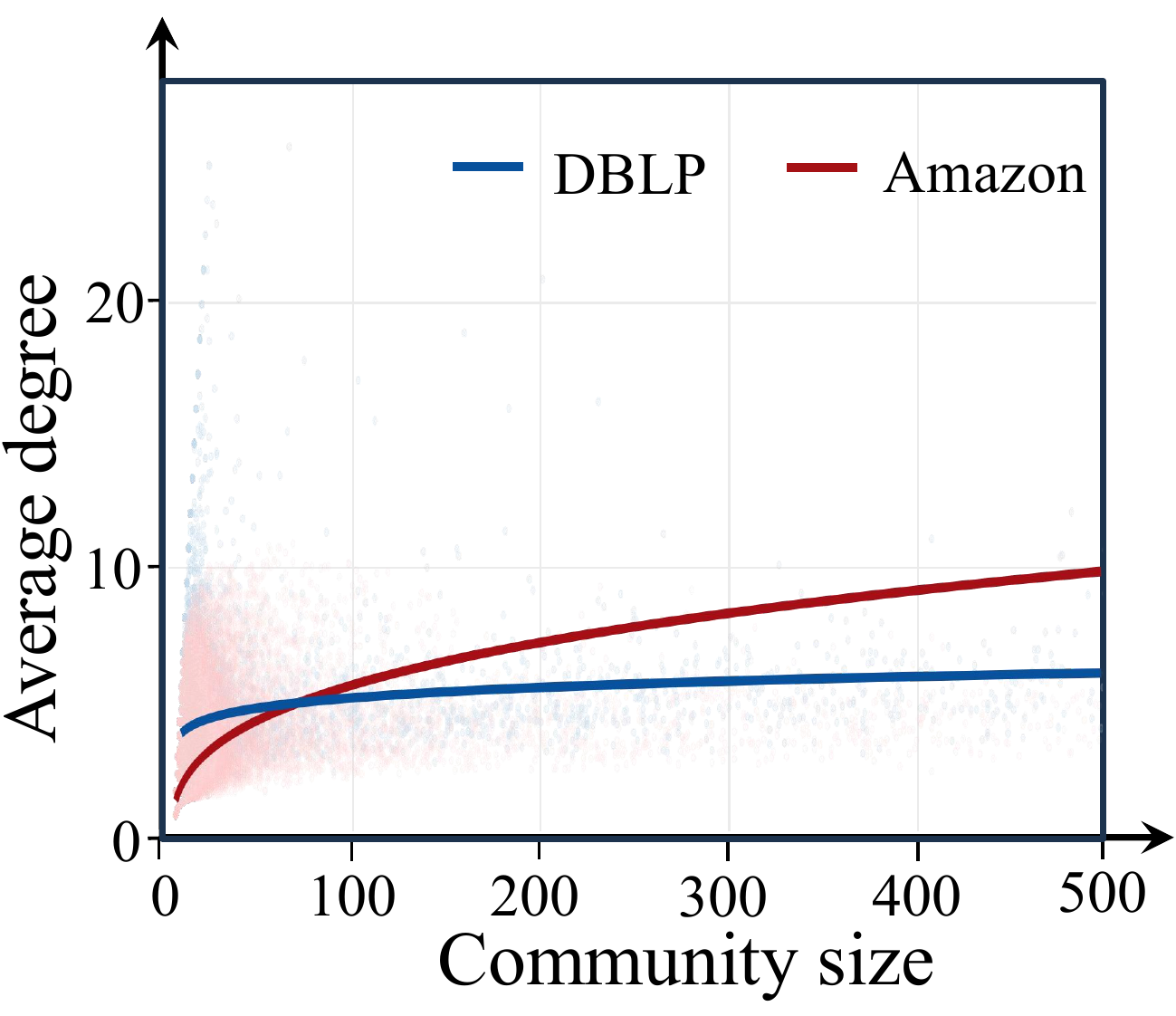}
    \caption{Social network tendency}
    \label{fig:motivation_intro1}
    \end{subfigure}
    \begin{subfigure}{.47\linewidth}
    \includegraphics[width=0.99\linewidth]{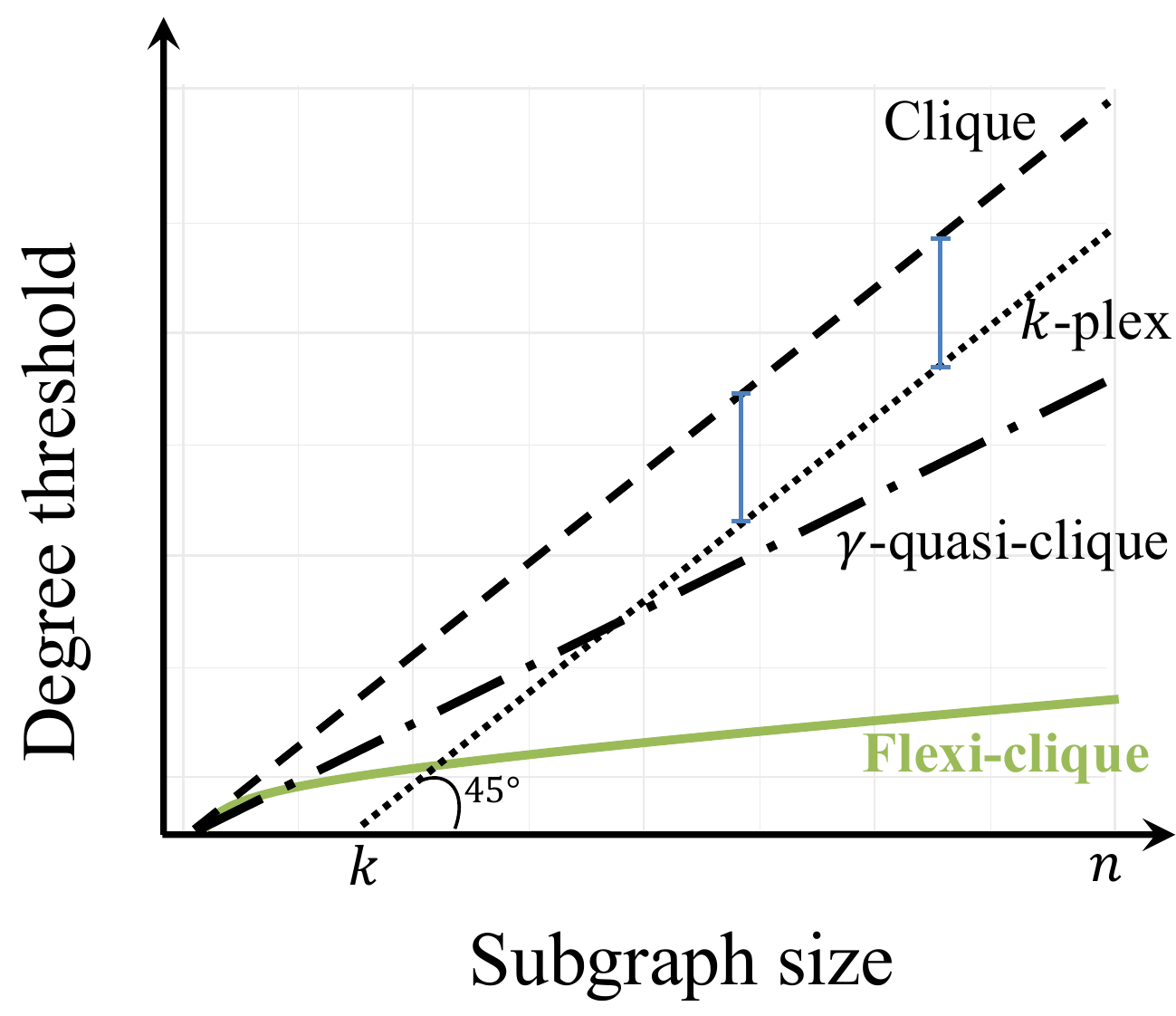}
    \caption{Minimum degree ratio}
    \label{fig:motivation_intro2}
    \end{subfigure}
\caption{Relationships between subgraph size and degree}
\label{fig:motivation_intro}
\vspace{-0.5cm}
\end{figure}

\spara{Design requirements for size-adaptive cohesion.}
Empirical evidence consistently shows that as communities grow, their internal connectivity decreases in a sub-linear manner. This phenomenon suggests that cohesive subgraph models intended for large-scale networks must satisfy three minimal requirements. First, cohesion should be enforced at the node level rather than as a global average, as density-based constraints may mask locally weakly supported nodes in large subgraphs. Second, cohesion requirements must increase with subgraph size while remaining sub-linear, allowing large structures to exist without enforcing near-clique density. Third, connectivity must be preserved to avoid fragmented solutions that lack community semantics. Existing models fail to satisfy these requirements simultaneously.

Figure~\ref{fig:motivation_intro1} illustrates the relationship between community size and average degree in real-world networks such as DBLP and Amazon. The observed pattern exhibits a clear sub-linear growth as community size increases. To relate this empirical tendency to model-imposed constraints, we examine the degree thresholds required by different cohesive subgraph definitions. Figure~\ref{fig:motivation_intro2} presents the corresponding minimum degree ratios implied by each model. Clique and $k$-plex enforce degree thresholds that grow linearly with subgraph size. Similarly, the $\gamma$-quasi-clique model follows a linear growth rule with a smaller slope $\gamma < 1$. Although these models differ in their tolerance to missing edges, they all rely on linear threshold functions with respect to subgraph size. As a result, they are unable to reflect the sub-linear connectivity decay that is consistently observed in large real-world communities.

\spara{Early work on the size-adaptive subgraphs.} To address this limitation, Kim et al.~\cite{flexi} introduced the {\flexi} model, which enforces a size-adaptive minimum-degree constraint of $\lfloor |H|^{\tau}\rfloor$ with $\tau\in[0,1)$. This formulation preserves node-wise stability while allowing cohesion requirements to relax smoothly with scale. Small subgraphs retain clique-like constraint, whereas larger ones remain cohesive without enforcing excessive degree. The parameter $\tau$ controls the sensitivity of cohesion to subgraph size, acting as a scale-sensitivity exponent rather than a fixed density threshold. As shown in Figure~\ref{fig:motivation_intro2}, the implied degree threshold of {\flexi} sub-linearly increases with size, aligning closely with empirical observations in Figure~\ref{fig:motivation_intro1}.

\spara{Non-heredity is inherent, not a limitation.}
A natural structural implication of adopting size-adaptive degree constraints is that the resulting model no longer satisfies hereditary or quasi-hereditary properties: a subgraph that satisfies the constraint does not necessarily preserve feasibility after the removal of a single node. Thus, this behaviour is not a modelling flaw, but an inherent consequence of enforcing sub-linear, size-dependent cohesion. This structural characteristic fundamentally distinguishes {\flexi} from classical clique relaxations and therefore necessitates a fundamentally different algorithmic treatment.

\spara{Limitations in the previous work.}
The initial study~\cite{flexi} demonstrated the empirical promise of the {\flexi} model, but was limited to heuristic exploration. It did not examine the algorithmic consequences of the size-dependent degree constraint, which breaks both hereditary and quasi-hereditary properties that underpin existing cohesive subgraph frameworks. As a result, no exact algorithm for the maximum {\flexi} problem was developed, fundamental questions regarding approximability remained unanswered, and the experimental analysis was confined to relatively small graphs. This work addresses these structural limitations directly. We provide the first formal analysis of the {\flexi} model, establishing its non-hereditary nature, inapproximability, and show why classical enumeration and optimisation techniques are inapplicable. 

Building on these insights, we develop an efficient heuristic and the first exact algorithm for the maximum {\flexi} search, based on a connectivity-preserving branching framework with six pruning rules. Extensive experiments on large real-world and synthetic networks demonstrate both effectiveness and practical efficiency.

\noindent\textbf{Contributions.} This work provides the first comprehensive theoretical and algorithmic study of the {\flexi} model. The main contributions are as follows:
\begin{itemize}[leftmargin=*]
    \item We propose an exact branch-and-bound algorithm for the maximum {\flexi} problem, together with an efficient heuristic that produces high-quality solutions and serves as a strong initial bound for exact search.
    \item We establish fundamental theoretical properties of the {\flexi} model, including inapproximability, and an explicit characterisation of its non-hereditary structure.
    \item We introduce a connectivity-preserving branching framework for non-hereditary cohesive subgraphs, supported by six pruning rules that enable effective and scalable exact computation.
    \item We perform extensive experimental evaluation on large real-world and synthetic networks, demonstrating both the effectiveness of the heuristic and the practical efficiency of the exact algorithm.
\end{itemize}

Overall, this study establishes a solid algorithmic and theoretical foundation 
for the {\flexi} model, positioning it as a practical and scalable framework for large-network analysis.

%% file: 02_preliminaries.tex
\section{Preliminaries}\label{sec:preliminaries}
This section introduces the notation and core definitions used in the problem formulation.
Let $G = (V, E)$ be a simple undirected graph, where $V$ and $E$ represent the sets of nodes and edges. We use $n = |V|$ and $m = |E|$ to denote the number of nodes and edges, respectively. For any set of nodes $H \subseteq V$, $G[H] = (H, E[H])$ represents the subgraph induced by $H$ where $E[H]$ is the set of edges with both endpoints in $H$. The set of neighbours of a node $v$ in graph $G$ is $N(v, G)$, which consists of all nodes adjacent to $v$ in $G$. For a set of nodes $H \subseteq V$, we denote $N_H(v, G)$ as the neighbours of $v$ in $G$ that belong to the set $H$, i.e., $N_H(v, G) = N(v, G) \cap H$. The degree of a node $v \in V$ in $G$ is $d(v, G)$ and represents $|N(v, G)|$. The minimum and maximum degrees of graph $G$ are $\delta(G)$ and $\Delta(G)$, respectively. The shortest path length between two nodes $u$ and $v$ in graph $G$ is $\operatorname{sp}_G(u,v)$. When the context is clear, we abbreviate these notations as $N(v)$, $N_H(v)$, $d(v)$, and $\operatorname{sp}(u,v)$.

We begin with two widely studied cohesive subgraph models: the $k$-core and the densest subgraph.

\begin{definition}[$k$-core~\cite{seidman1983network}]
Given a graph $G = (V, E)$ and a positive integer $k$, a $k$-core $C_k$ is a maximal set of nodes $H \subseteq V$ such that the induced subgraph $G[H]$ has minimum degree at least $k$, \textit{i.e.}, $\delta(G[H]) \geq k$. 
\end{definition}

The $k$-core structure is uniquely defined for any graph and forms a nested hierarchy, where each $(k+1)$-core is contained within the $k$-core~\cite{ malliaros2020core}. It is known that finding $k$-core structure takes $O(m)$ time~\cite{batagelj2003m}. 

\begin{definition}[Densest subgraph~\cite{lanciano2024survey}]
Given a graph $G = (V, E)$, the densest subgraph is a subgraph $G[H] = (H, E[H])$ that maximises the density defined as $\frac{|E[H]|}{|H|}$, where $|E[H]|$ and $|H|$ are the numbers of edges and nodes in $G[H]$.
\end{definition}

The densest subgraph problem admits a polynomial-time solution via maximum flow~\cite{goldberg1984finding}. A well-known $2$-approximation algorithm iteratively removes the minimum-degree node and tracks the maximum intermediate density~\cite{charikar2000greedy}.

These models are adopted as algorithmic primitives due to their ability to guarantee minimum degree (via $k$-core) and identify dense regions (via densest subgraph) efficiently in polynomial time. We next review several cohesive subgraph models relevant to our study. 

\begin{definition}[$k$-clique]
Given a graph $G = (V, E)$ and a positive integer $k$, a $k$-clique is a set of nodes $H \subseteq V$ with $|H| = k$, where every pair of nodes is connected in the induced subgraph $G[H]$: $\forall u, v \in H, (u, v) \in E$.
\end{definition}

The $k$-clique problem is $\mathrm{NP}$-hard~\cite{karp2010reducibility}, with no constant-factor approximation unless $\mathrm{P} = \mathrm{NP}$~\cite{lin2021constant}. Also, its strict edge requirement often limits applicability in real-world networks~\cite{conte2017fast,wang2022listing}.

To mitigate this, several relaxed clique models have been introduced. The $k$-plex model~\cite{seidman1978graph} allows each node to miss up to $k$ connections, requiring each node $v$ to have at least $|H| - k$ neighbours within $G[H]$. Two variants of the $\gamma$-quasi-clique model relax the constraints by requiring either (i) that each node has at least $\lceil \gamma \cdot (|H|-1)\rceil$ neighbours, or (ii) that the subgraph maintains a density of at least $\gamma$~\cite{sanei2021mining,tsourakakis2013denser}. Both quasi-clique variants reduce to a clique when $\gamma = 1$. These models benefit from key structural properties that enable efficient algorithmic design. The hereditary property ensures that any induced subgraph of a cohesive subgraph also satisfies the model constraints, which both clique and $k$-plex possess~\cite{chang2022efficient}. The quasi-hereditary property guarantees that any qualifying subgraph of size $k$ contains at least one proper subset of size $k-1$ that also satisfies the constraints, as exhibited by density-based $\gamma$-quasi-cliques~\cite{ribeiro2019exact}. These properties are routinely utilised for finding maximal or maximum cohesive subgraphs~\cite{trukhanov2013algorithms,chang2024maximum,mahdavi2014branch}. In the following section, we show that the proposed {\flexi} model lacks both hereditary and quasi-hereditary properties, thereby motivating the need for a fundamentally new algorithmic approach.

%% file: 03_problem_statement.tex
\section{PROBLEM STATEMENT}\label{sec:problem_statement}
This section introduces the concept of the {\flexi} model~\cite{flexi}, formally defines the corresponding maximisation problem, and discusses its properties. The {\flexi} model offers a degree-based relaxation of the clique structure, enabling the discovery of cohesive subgraphs over a spectrum of graph densities. Unlike the $k$-plex and $\gamma$-quasi-clique models, whose degree constraints grow linearly with subgraph size, {\flexi} employs a sub-linear growth rule.

\begin{definition}[{\flexi}~\cite{flexi}]\label{def:flexi}
Given a graph $G = (V, E)$ and a parameter $\tau \in [0,1)$, a {\flexi}, denoted as $H$, is a set of nodes that satisfies the following properties:
\begin{itemize}[leftmargin=*]
    \item Every node $u \in H$ has at least $\lfloor |H|^\tau \rfloor$ neighbours in the induced subgraph $G[H]$, \textit{i.e.}, $\delta(G[H]) \geq \lfloor |H|^\tau \rfloor$.
    \item The induced subgraph $G[H]$ is connected.
\end{itemize}
\end{definition}

The main goal of {\flexi} analysis is to find the largest subgraph satisfying the definition. We formalise this as follows.

\begin{problemDefinition}
[Maximum {\flexi} problem~\cite{flexi}] Given a graph $G = (V, E)$ and a parameter $\tau \in [0,1)$, the maximum {\flexi} problem is to find the largest {\flexi} of the given graph.
\end{problemDefinition}

Although the largest {\flexi} may not be unique, its size is uniquely determined. To establish the computational complexity of the maximum {\flexi} problem, we consider its decision version and provide a reduction from the clique decision problem. This direct reduction demonstrates the $\mathrm{NP}$-hardness of the maximum {\flexi} problem.

\begin{problemDefinition}[{\flexi} decision problem]
\label{def:flexiclique_decision}
Given a graph $G=(V,E)$, a parameter $\tau \in [0,1)$, and a positive integer $k$, {\flexi} decision problem is to decide whether there exists a {\flexi} $H \subseteq V$ in $G$ such that $|H| \geq k$.
\end{problemDefinition}

\begin{theorem}
\label{lem:kflexi_nphard}
The {\flexi} decision problem is $\mathrm{NP}$-hard.
\end{theorem}

The $\mathrm{NP}$-hardness of the {\flexi} decision problem was established by Kim et al.~\cite{flexi} via a reduction from the clique decision problem. Since this result is not the main focus of the present work, we omit the proof in the main body and provide it in the Appendix A~\cite{kim2026appendix} for completeness.

Next, we establish a fundamental limitation of the Maximum {\flexi} problem by showing that it does not admit any constant-factor approximation. 
This result complements the NP-hardness established above and clarifies that the computational difficulty of {\flexi} is not merely due to exact optimisation, but also persists under approximation.

\begin{theorem}
The Maximum {\flexi} problem does not admit a constant-factor approximation unless $\mathrm{P}=\mathrm{NP}$.
\end{theorem}

\begin{proof}
We prove the claim via a gap-preserving reduction from the gap version of the \textsc{Clique} problem.
Specifically, given a graph $G=(V,E)$, an integer $K$, and a constant $c>1$, the task is to distinguish between the following two cases, under the promise that one of them holds:
\begin{itemize}[leftmargin=*]
    \item \textbf{YES}: $G$ contains a clique of size at least $K$;
    \item \textbf{NO}: every clique in $G$ has size at most $K/c$.
\end{itemize}
It is NP-hard to distinguish between the above two cases~\cite{feige1991approximating}.

Given a graph $G=(V,E)$ with parameters $K$ and $c>1$, under the promise that $G$ contains a clique of size at least $K$ or that every clique in $G$ has size at most $K/c$, we construct a graph $G'=(V',E')$ that preserves the promised gap while amplifying differences in clique size. Specifically, for each node $v \in V$, we create a clique $C_v$ consisting of exactly $B$ nodes, where $B$ is a sufficiently large polynomial in $n$. The node set of $G'$ is defined as $V'=\bigcup_{v\in V} C_v$, and we denote $|V'|$ by $N$.
For any two distinct nodes $u,v \in V$, the cliques $C_u$ and $C_v$ are made fully adjacent if $(u,v) \in E$, and completely non-adjacent otherwise.
Next, we set the parameter $\tau$ as
\begin{align*}
\tau = 1-\frac{1}{4n\log N},
\end{align*}
which will be used to control the gap between $s$ and $s^\tau$ in the analysis below. The graph construction process for the gap-preserving reduction
is illustrated in Figure~\ref{fig:reduction}.
For the constructed graph $G'$, we denote by $\mathrm{OPT}(G')$ the maximum size of a {\flexi}.

\begin{figure}[t]
\centering
\includegraphics[width=0.85\linewidth]{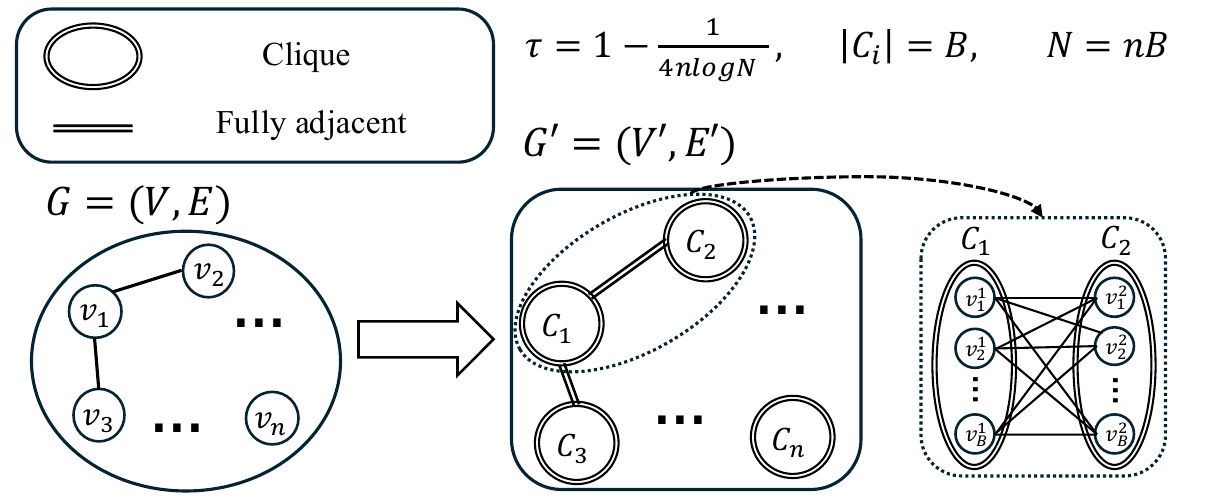}
\caption{Reduction process} 
\vspace{-0.5cm} 
\label{fig:reduction}
\end{figure}

\spara{YES Case.}
Assume that $G$ contains a clique $Q\subseteq V$ with $|Q|=K$, and let $S = \bigcup_{v\in Q} C_v \subseteq V'$. 
Clearly, the set $S$ has size $KB$, and the induced subgraph $G'[S]$ is connected since all cliques corresponding to nodes in $Q$ are adjacent. Consider any node $x\in S$, since $C_v$ is a clique, $x$ has $B-1$ neighbours inside $G'[C_v]$. Moreover, from the other $K-1$ cliques in $S$, the node $x$ has in total $(K-1)B$ neighbours. Therefore,
\begin{align}
\deg_{G'[S]}(x) \ge (B-1)+(K-1)B = KB-1.
\end{align}
As $\tau\in [0,1)$, $KB-1 \ge \lfloor (KB)^\tau \rfloor$
for the chosen value of $B$.
Hence $\deg_{G'[S]}(x)\ge \lfloor |S|^\tau\rfloor$ for all $x\in S$, and thus $S$ is a valid {\flexi} with $|S|=KB$, implying
\begin{align}\label{eq:eq2}
\mathrm{OPT}(G') \ge KB.
\end{align}


\spara{NO Case.} \textcolor{black}{
Now assume that every clique in $G$ has size at most $K/c$. Let $S\subseteq V'$ be an arbitrary {\flexi} in $G'$. For each $v\in V$, define $a_v = |S\cap C_v|$, and let
\begin{align}
U = \{v\in V : a_v > 0\}.
\end{align}
Denote
\begin{align}
s = |S| = \sum_{v\in U} a_v.
\end{align}
}
\textcolor{black}{
We first bound the total number of non-neighbours that a node in a {\flexi} can have. Recall that $\tau = 1-\frac{1}{4n\log N}$ and that $s \le nB = N$. Thus,
\begin{align}
s - s^\tau \le \frac{s}{4n} \le \frac{B}{4}.
\label{eq:non-neigh-bound}
\end{align}
}
\textcolor{black}{
Fix any $u\in U$ and consider an arbitrary node $x\in S\cap C_u$. By construction of $G'$, the degree of $x$ in the induced subgraph $G'[S]$ is
\begin{align}
\deg_{G'[S]}(x)
= (a_u - 1) + \sum_{\substack{v\in U:\\ (u,v)\in E}} a_v.
\end{align}
Therefore, the number of nodes in $S$ that are non-adjacent to $x$ is
\begin{align}
(s-1) - \deg_{G'[S]}(x)
= \sum_{\substack{v\in U:\\ (u,v)\notin E}} a_v.
\end{align}
}
\textcolor{black}{
Since $S$ is a {\flexi}, we have $\deg_{G'[S]}(x) \ge \lfloor s^\tau \rfloor$, and hence
\begin{align}
\sum_{\substack{v\in U:\\ (u,v)\notin E}} a_v
&\le (s-1) - \lfloor s^\tau \rfloor 
\le s - s^\tau \le \frac{B}{4},
\label{eq:block-nonneigh}
\end{align}
where the last inequality follows from~\eqref{eq:non-neigh-bound}.
}
\textcolor{black}{
Now let $C$ be a maximum clique in the induced graph $G[U]$. By the NO-case assumption, we have $|C| \le K/c$. By maximality of $C$, for every $v\in U\setminus C$ there exists a node $u\in C$ such that $(u,v)\notin E$, and we assign each $v\in U\setminus C$ to one such $u$.
}
\textcolor{black}{
For any fixed $u\in C$, it follows from~\eqref{eq:block-nonneigh} that
\begin{align}
\sum_{\substack{v\in U\setminus C}} a_v
\le
\sum_{\substack{v\in U:\\ (u,v)\notin E}} a_v
\le \frac{B}{4}.
\end{align}
Summing over all $u\in C$, we obtain
\begin{align}
\sum_{v\in U\setminus C} a_v \le |C|\cdot \frac{B}{4}.
\label{eq:outside-clique}
\end{align}
}
\textcolor{black}{
Since each block $C_v$ has size exactly $B$, we also have
\begin{align}
\sum_{u\in C} a_u \le |C|B.
\end{align}
Combining this with~\eqref{eq:outside-clique}, we conclude that
\begin{equation}
\begin{aligned}
|S|
&= \sum_{u\in C} a_u + \sum_{v\in U\setminus C} a_v \\
&\le |C|B + |C|\cdot \frac{B}{4}
= \frac{5}{4}|C|B
\le \frac{5}{4}\cdot \frac{K}{c}B .
\end{aligned}
\label{eq:no-case-bound}
\end{equation}
Therefore,
\begin{align}\label{eq:eq_fin}
\mathrm{OPT}(G') \le \frac{5}{4}\cdot \frac{K}{c}B.
\end{align}
In conclusion, combining~\eqref{eq:eq2} and~\eqref{eq:eq_fin}
yields a constant gap between \textbf{YES} and \textbf{NO} instances.
Therefore, the Maximum {\flexi} problem does not admit any
constant-factor approximation unless $\mathrm{P}=\mathrm{NP}$.
}
\end{proof}

Beyond its computational hardness, we next analyse the structural properties of {\flexi}, showing that it lacks both the hereditary and quasi-hereditary properties commonly assumed in cohesive subgraph models.

\begin{property}[Non-hereditary]
{\flexi} is non-hereditary.
\end{property}
\begin{proof}
Let $H \subseteq V$ be a {\flexi} subgraph with minimum degree $\delta(G[H]) = \lfloor |H|^\tau \rfloor$, and let $x \in H$ be a neighbour of the lowest-degree node $v \in H$. If $H \setminus \{x\}$ is also a {\flexi}, it must satisfy $\delta(G[H \setminus \{x\}]) \geq \lfloor (|H|-1)^\tau \rfloor$. However, the removal of $x$ may reduce the degree of $v$ by one, possibly violating the threshold, especially when $\lfloor |H|^{\tau} \rfloor = \lfloor (|H|-1)^{\tau} \rfloor$. This contradicts the assumption, thus {\flexi} is not hereditary.
\end{proof}

\begin{figure}[t]
\centering
\small
\begin{subfigure}{.35\linewidth}
\includegraphics[width=0.99\linewidth]{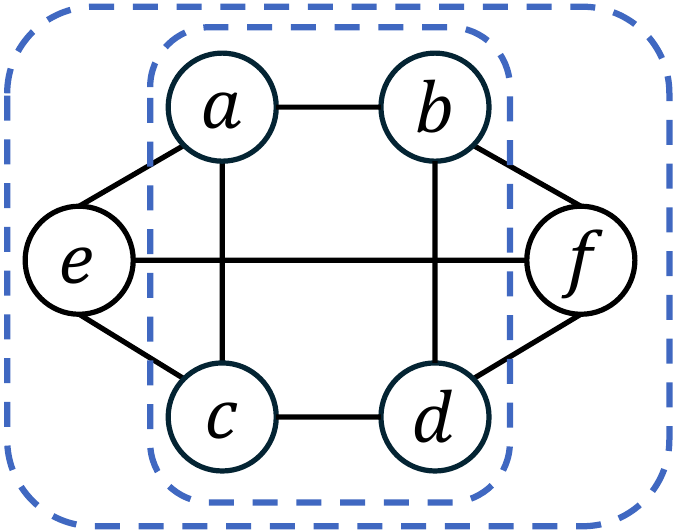}
\caption{{\flexi}}
\label{fig:non_hereditary1}
\end{subfigure}
\hspace{0.3cm}
\begin{subfigure}{.35\linewidth}
\includegraphics[width=0.99\linewidth]{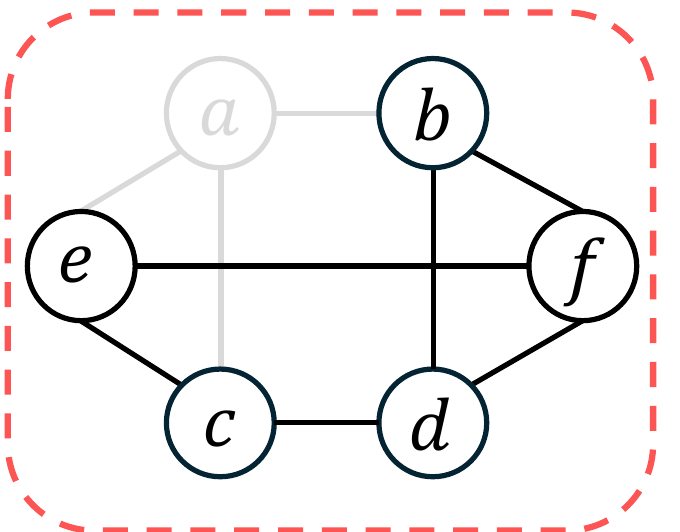}
\caption{non-{\flexi}}
\label{fig:non_hereditary2}
\end{subfigure}
\vspace{-0.2cm}
\caption{Non-hereditary property}
\label{fig:non_hereditary}
\vspace{-0.5cm}
\end{figure}

\begin{property}[Non-quasi-hereditary]
{\flexi} is non-quasi-hereditary.
\end{property}
\begin{proof}
Suppose $H \subseteq V$ is a {\flexi} with $\delta(G[H]) = \Delta(G[H])$. For any $x \in H$, removing $x$ decreases the degree of its neighbours, potentially violating the condition $\delta(G[H \setminus \{x\}]) \geq \lfloor (|H| - 1)^\tau \rfloor$, even if $\lfloor |H|^\tau \rfloor = \lfloor (|H| - 1)^\tau \rfloor$. Thus, {\flexi} is not quasi-hereditary.
\end{proof}

Having established the $\mathrm{NP}$-hardness of the maximum {\flexi} problem and its non-hereditary nature, we conclude that dedicated algorithmic strategies are required to address the challenges unique to {\flexi}.

\begin{example}
Figure~\ref{fig:non_hereditary} illustrates the non-hereditary properties of the {\flexi}. For the given graph in Figure~\ref{fig:non_hereditary1} with $\tau = 0.75$, the set of nodes $\{a,b,c,d,e,f\}$ and $\{a,b,c,d\}$ form the {\flexi}. However, removing any single node from the graph results in a set of nodes that fails to meet the constraint (as $\lfloor 5^{0.75} \rfloor = 3$). This demonstrates that {\flexi} is both non-hereditary and non-quasi-hereditary, and implies that the absence of an $n$-sized {\flexi} does not preclude the existence of a $(n+1)$-sized or even larger one. 
\end{example}

%% file: 04_algorithm_heuristics.tex
\section{Flexi-Prune Algorithm (\FPA)}\label{sec:algorithm_heuristic}

While exact algorithms provide optimal solutions, they often become impractical for large graphs due to their high computational cost. To achieve scalability, we propose the Flexi-Prune Algorithm (\FPA), a fast heuristic designed to produce high-quality solutions efficiently. {\FPA} integrates two key design principles: (i) \textit{core-based seed selection}, which locates dense regions likely to contain feasible {\flexi} subgraphs, and (ii) \textit{connectivity-aware pruning}, which incrementally refines the subgraph while preserving structural coherence.

\spara{High-level idea.} {\FPA} proceeds in two stages. First, it identifies a promising seed by progressively increasing $k$ in the $k$-core hierarchy until the largest $k$-core component satisfies the {\flexi} degree constraint, and then confines the search to the enclosing $(k-1)$-core. This step ensures that the algorithm starts from a dense yet sufficiently large region. In the second stage, {\FPA} iteratively removes the lowest-degree node that is not an articulation point(i.e., its removal does not disconnect the subgraph), updating node degrees and the threshold after each removal until all remaining nodes satisfy the {\flexi} condition. This pruning process effectively eliminates peripheral nodes while maintaining connectivity, yielding a cohesive subgraph that closely approximates the maximum {\flexi}.

\spara{Seed selection via core decomposition.} Because {\flexi} requires each node to satisfy a size-dependent degree threshold, dense core regions serve as natural candidates for feasible solutions. Real-world networks generally follow a power-law degree distribution (Fig.~\ref{fig:observation1}), and thus many peripheral low-degree nodes cannot meet this requirement. Furthermore, within each $k$-core, the largest connected component (LCC) dominates the whole structure in terms of the size, as shown in Fig.~\ref{fig:observation2}, where the LCC ratio, i.e., $|$LCC$|$ / $|k$-core$|$, remains close to $1$ even as $k$ increases. This observation implies that most nodes in a core are contained within the single largest connected component, which can serve as an effective starting point for exploration. 

Also, increasing $k$ gradually removes low-degree nodes, yielding smaller yet denser subgraphs; at the same time, the required degree threshold $\lfloor |H|^{\tau}\rfloor$ decreases. {\FPA} therefore increments $k$ until the LCC of the $k$-core first satisfies the {\flexi} constraint and then confines the search to the enclosing $(k-1)$-core. If no such $k$ is found, {\FPA} starts from the LCC of the $k$-core with the largest $k$ value, ensuring that the process begins in the densest region available. This strategy ensures that the algorithm begins from a dense but not overly restrictive region, balancing efficiency and coverage of potential maximal {\flexi} subgraphs.


\begin{figure}[t]
\centering
\small
\captionsetup[subfigure]{
  justification=centering,
  margin={0.4cm,0cm}
}
\begin{subfigure}[b]{0.48\linewidth}
\centering
\includegraphics[width=0.99\linewidth]{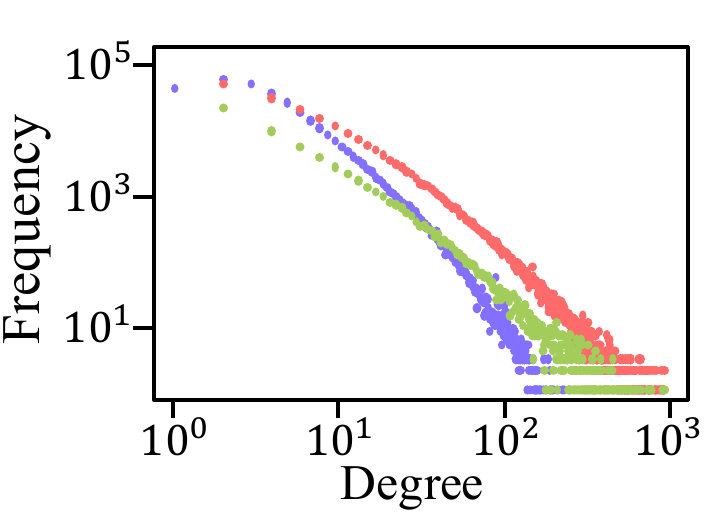}
\vspace{-0.3cm}
\caption{Degree distribution}
\label{fig:observation1}
\end{subfigure}
\begin{subfigure}[b]{0.48\linewidth}
\centering
\includegraphics[width=0.99\linewidth]{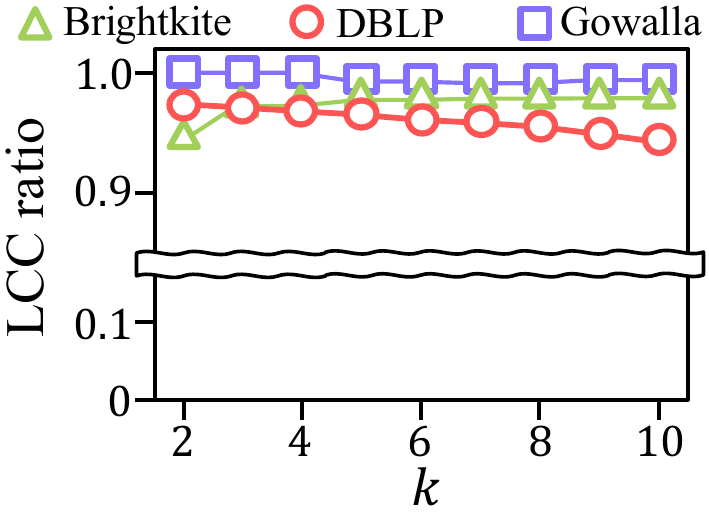}
\vspace{-0.3cm}
\caption{LCC ratio}
\label{fig:observation2}
\end{subfigure}
\caption{Observations on core decomposition}
\label{fig:observation}
\vspace{-0.7cm}
\end{figure}

\SetAlgoNoEnd
\begin{algorithm}[t]
\small
\SetAlgoLined
\SetKw{break}{break}
\SetKw{return}{return}
\SetKwData{false}{False}
\SetKwData{true}{True}
\SetKwData{null}{null}
\SetKwData{LCC}{LCC}
\SetKwFunction{buildFDC}{buildFDC}
\SetKwFunction{sort}{sortByDegree}
\SetKwFunction{emptyQueue}{emptyQueue}
\SetKwFunction{queue}{Queue}
\SetKwFunction{minDegree}{minDegree}
\SetKwFunction{nextNode}{nextNode}
\SetKwFunction{remove}{remove}
\SetKwFunction{degreeRearrange}{degreeRearrange}
\SetKwFunction{isArticulation}{isArticulation}
\SetKwIF{If}{ElseIf}{Else}{if}{:}{else if}{else}{}
\SetKwFor{While}{while}{:}{}

\KwIn{Graph $G=(V,E)$, parameter $\tau$}
\KwOut{{\flexi} $C$}

Compute $k$-core decomposition of $G$\;
$k^* \leftarrow \min\{k : \LCC \text{ of } k\text{-core satisfies {\flexi}}\}$\;
\If{$k^* = $ \null}{
    $C \leftarrow \LCC$ of the $k$-core with maximum $k$\;
}
\Else{
    \mbox{$C \leftarrow$ $(k^*-1)$-core component containing $\LCC$ of $k^*$-core\;}
}

\sort{$C$}\; 
\While{$C \neq \emptyset$}{
    \If{$\delta$ $(G[C])$ $\geq \lfloor |C|^\tau \rfloor$}{
        \return $C$\;
    }

    \For{$u \in C$}{
        \If{\isArticulation{$G[C], u$} = \false}{
            $C \leftarrow C \setminus \{u\}$\;
            \degreeRearrange{$N(u,G[C])$}\;
            \break\;
        }
    }
}
\return $C$\;
\caption{\mbox{\FuncSty{{\FPA}}: Flexi-Prune Algorithm}}
\label{alg:FPA}
\end{algorithm}

\spara{Connectivity-aware node pruning.} 
After the initial seed is obtained, {\FPA} performs iterative pruning while maintaining connectivity. To efficiently verify connectivity during node removal, we leverage a fully dynamic connectivity structure based on the multi-level spanning forest framework of Holm-de Lichtenberg-Thorup~\cite{holm2001poly}, 
in which each spanning forest can be represented using a link-cut tree~\cite{sleator1981data}, supporting \textsc{link} and \textsc{cut} operations in logarithmic time. Since the pruning phase always maintains the subgraph connected, we only need to test whether removing a node $u$ would disconnect the remaining graph. We therefore check connectivity by temporarily deleting all edges incident to $u$ and verifying whether all neighbours of $u$ belong to the same connected component in the resulting graph; if this condition is violated, the deleted edges are restored. This enables connectivity-preserving pruning without recomputing connected components at each iteration.

\spara{Algorithmic procedure.} The pseudo description of {\FPA} is shown in Algorithm~\ref{alg:FPA}. The algorithm first performs $k$-core decomposition and identifies the smallest $k^*$ whose largest connected component (LCC) satisfies the {\flexi} condition. If no such $k^*$ exists, it starts from the LCC of the $k$-core with the maximum $k$ value, ensuring that the search begins in the densest region. Nodes are then sorted in ascending order of degree. During the iterative peeling phase, {\FPA} repeatedly examines the minimum-degree node: if it is not an articulation point, it is removed and all affected degrees are updated; otherwise, the next node in order is tested. This process continues until the remaining subgraph satisfies the {\flexi} constraint, yielding a feasible structure.

\spara{Time complexity analysis.} The overall complexity of {\FPA} is dominated by three components. 
(i) \textit{Core decomposition} requires $O(m)$ time~\cite{batagelj2003m}. 
(ii) \textit{Iterative component extraction} for the $k$-core hierarchy costs $O(k^*(n+m))$, where connected components are computed in $O(n+m)$ per level. 
(iii) \textit{Peeling} takes $O(\Delta(G)\,n^2\log n)$, as each node removal may trigger repeated articulation checks on $O(n)$ nodes, and every dynamic connectivity operation requires $O(\log n)$ update cost per edge. 
Since $m\le\Delta(G)\,n$, the total complexity is $O(\Delta(G)\,n^2\log n)$, dominated by the peeling phase. 
In practice, real networks are sparse ($\Delta(G)\!\ll\!n$), and articulation point checks tends to occur $O(n)$ times rather than $O(n^2)$, leading to near-linear scaling as shown in Section~\ref{sec:experiments}.

\begin{figure}[t]
\centering
\small
\begin{subfigure}{.45\linewidth}
\includegraphics[width=0.99\linewidth]{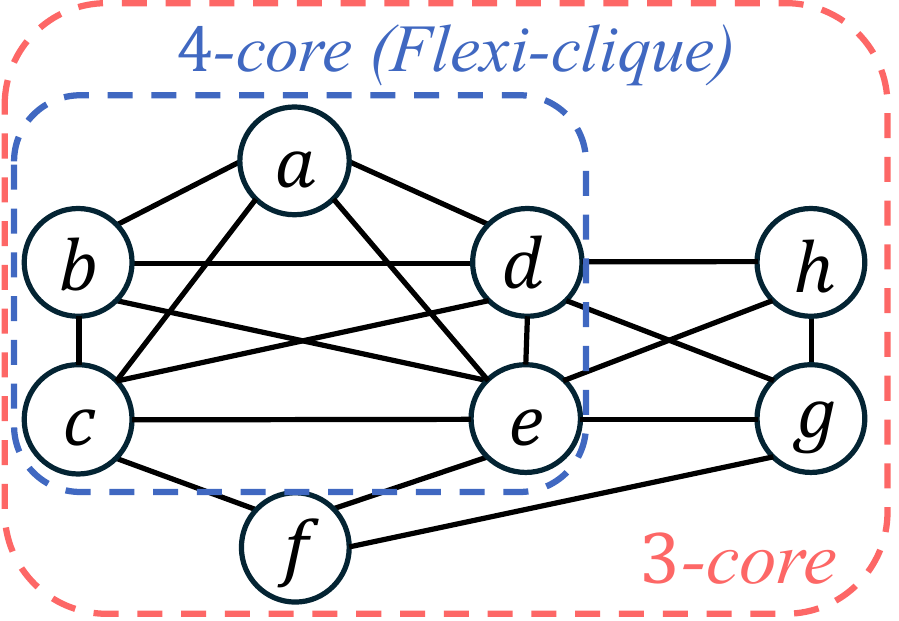}
\caption{Largest feasible core}
\label{fig:FPA1}
\end{subfigure}
\begin{subfigure}{.45\linewidth}
\includegraphics[width=0.99\linewidth]{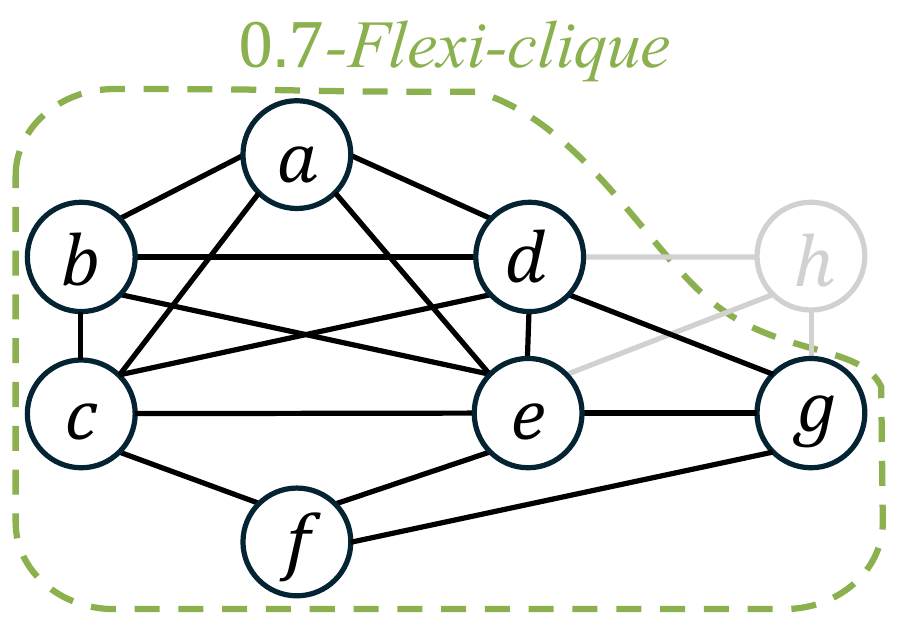}
\caption{Peeling process}
\label{fig:FPA2}
\end{subfigure}
\caption{{\FPA} example $(\tau = 0.75)$}
\label{fig:FPA}
\vspace{-0.2cm} 
\end{figure}

\begin{example}

Figure~\ref{fig:FPA} illustrates the execution of Algorithm~\ref{alg:FPA}. 
The core decomposition identifies the set $\{a,\dots,e\}$ as the $4$-core, which is feasible because $\lfloor 5^{0.75} \rfloor = 3$ and every node in this subgraph meets the {\flexi} degree constraint. In contrast, the $3$-core, consisting of $\{a,\dots,g\}$, violates the condition since at least one node has degree below $4$ while $\lfloor 7^{0.75}\rfloor = 4$. Consequently, the algorithm begins pruning from the $3$-core that contains the feasible {\flexi}. As shown in Figure~\ref{fig:FPA2}, nodes that are not articulation points are removed in ascending order of degree, and the degree constraint is re-evaluated after each deletion. In this example, removing node $f$ (degree $3$) yields a $6$-node subgraph satisfying $\lfloor 6^{0.75}\rfloor = 3$, confirming that the resulting structure is a valid {\flexi}. This example demonstrates how {\FPA} efficiently converges to a feasible subgraph by eliminating peripheral nodes while preserving connectivity.
\end{example}

%% file: 05_algorithm_BB.tex
\begin{table}[t]
\centering
\small
\setlength{\tabcolsep}{6pt}
\caption{Algorithmic challenges across models}
\label{tab:model_comparison1}
\resizebox{\columnwidth}{!}{%
\begin{tabular}{c|c|c|c|c}
\hline
\textbf{Model}
& \textbf{Non-Her.}
& \textbf{Non-Q-Her.}
& \textbf{Conn.}
& \textbf{No Diam. Bd.} \\
\hline\hline
Clique~\cite{bron1973algorithm}
& $\times$ & $\times$ & $\times$ & $\times$ \\ \hline

$k$-plex~\cite{chang2024maximum}
& $\times$ & $\times$ & $\triangle$ & $\triangle$ \\ \hline

QC-DT~\cite{mahdavi2014branch}
& $\bigcirc$ & $\times$ & $\triangle$ & $\triangle$ \\ \hline

QC-DG~\cite{yu2023fast}
& $\bigcirc$ & $\times$ & $\triangle$ & $\triangle$ \\ \hline

Flexi-clique
& $\bigcirc$ & $\bigcirc$ & $\bigcirc$ & $\bigcirc$ \\ \hline
\end{tabular}
}
\vspace{-0.5cm}
\end{table}

\section{Efficient Branch and Bound Algorithm ({\EBA}) }\label{sec:algorithm_BB}

To obtain an exact solution for the maximum {\flexi} problem, we devise an Efficient Branch-and-bound Algorithm (\EBA) that systematically explores the solution space while aggressively pruning unpromising regions to avoid redundant exploration; unlike heuristic methods, {\EBA} guarantees optimality by exhaustive yet structured search under connectivity constraints. At each \textit{search node}  (i.e., a node of the branch-and-bound tree), the algorithm maintains four mutually disjoint sets of graph nodes that organise branching and prevent reconsideration of excluded nodes:

\begin{itemize}[leftmargin=*]
    \item $S$ (partial set): nodes already selected in the current search node: $S$ must induce a connected subgraph, but does not necessarily form a valid {\flexi}.
    \item $C^r$ (reachable candidates): nodes not in $S$ but adjacent to $S$, thus eligible to extend $S$ without breaking connectivity,
    \item $C^{un}$ (unreachable candidates): nodes currently non-adjacent to $S$ and therefore not immediately addable; these may become reachable after $S$ expands, 
    \item $D$ (exclusion set): nodes explicitly ruled out in the current search node and all of its descendants.
\end{itemize}

\subsection{Connectivity-Preserving Branching ({\CP})}\label{sec:subsec5_1} 

Existing cohesive subgraph models typically satisfy structural assumptions such as (quasi-)heredity.  Moreover, connectivity and small diameter are often implicitly guaranteed once the model parameters are fixed. For example, a $k$-plex is guaranteed to be connected when $|H| \ge 2k-1$, and a $\gamma$-quasi-clique is connected when $\gamma \ge 0.5$, and moreover their solutions both have diameter at most $2$ under the same condition~\cite{chang2024maximum,yu2023fast}. Table~\ref{tab:model_comparison1} summarises these algorithmic characteristics across different cohesive subgraph models. Here, $\bigcirc$, $\triangle$, and $\times$ indicate that a given characteristic is intrinsic to the model, arises only under additional parameter conditions, or is absent, respectively.

These properties allow conventional branching frameworks to safely prune the search space: feasibility is preserved under candidate extensions, and the search can be restricted to a local neighbourhood. However, the {\flexi} model violates all of these assumptions. It is neither hereditary nor quasi-hereditary, and admits no parameter-dependent subgraph-size threshold that guarantees connectivity or bounded diameter. As a result, conventional branching strategies may generate disconnected partial sets or fail to explore valid candidates.

To address this challenge, we employ a novel branching strategy called \emph{Connectivity-Preserving Branching} (\CP). The proposed {\CP} framework is explicitly designed for such settings explained above. It maintains connectivity throughout the search by distinguishing reachable candidates ($C^r$) from currently unreachable ones ($C^{un}$), ensuring that every partial solution induces a connected subgraph. This connectivity-aware design enables exhaustive yet safe exploration of the search space and forms the structural basis of {\EBA}. The formal process is defined as follows.

\begin{figure}[t]
\centering
\includegraphics[width=0.9\linewidth]{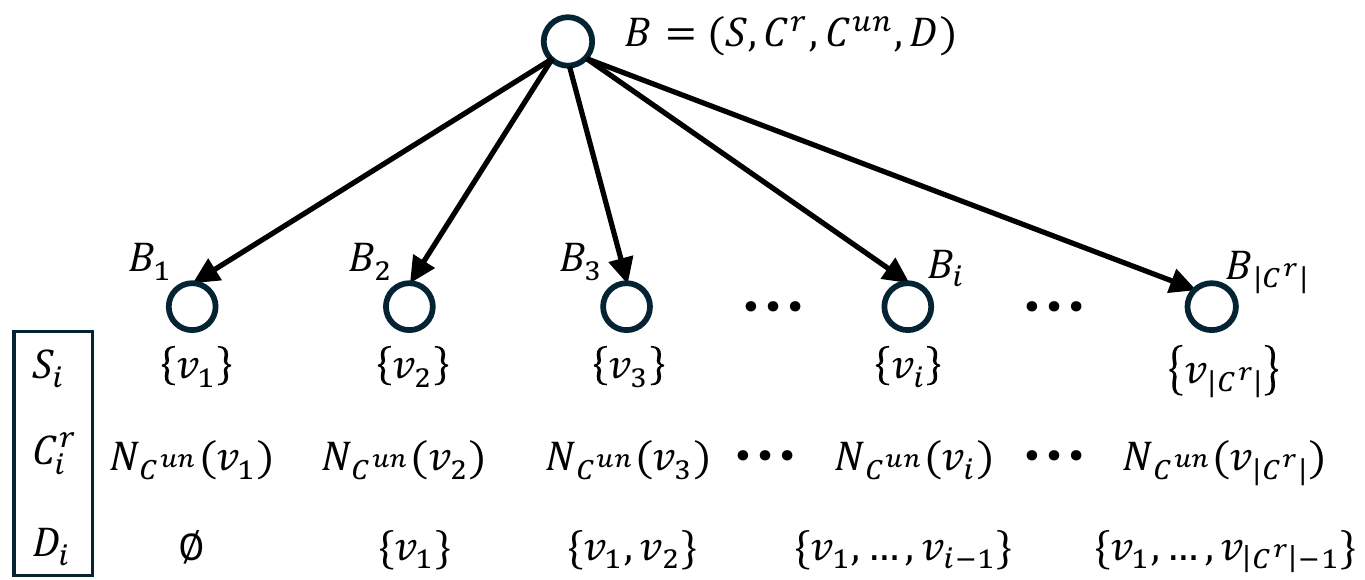}
\caption{Branching process}
\vspace{-0.5cm}
\label{fig:branching}
\end{figure}

For a given graph $G = (V,E)$, the search begins from the root search node $B_0 = (\emptyset, \emptyset, V, \emptyset)$. Each search node $B = (S, C^r, C^{un}, D)$ represents a distinct state of the search tree, where $S$ denotes the current partial set, $C^r$ and $C^{un}$ are the reachable and unreachable candidate sets, and $D$ is the exclusion set. During expansion, a child node is created by selecting a node $v$ to extend $S$. If $S$ is empty (i.e., the root), $v$ is chosen from $C^{un}$; otherwise, $v$ must be selected from $C^r$ to preserve connectivity. In both cases, $v$ must not belong to $D$ so that every node in $G$ is considered  once within the current search node. For given search node $B$, let $v_i$ be the $i$-th valid node satisfying these conditions, the child search node $B_i = (S_i, C^r_i, C^{un}_i, D_i)$ is constructed as follows:
\begin{center}
\vspace{-0.6cm}
\begin{align*} \label{eq:branching}
S_i &= S \cup \{v_i\} \\
C^r_i &= (C^r \setminus \{v_1,\dots,v_i\}) \cup N_{C^{un}}(v_i,G) \\
C^{un}_i &= C^{un} \setminus N_{C^{un}}(v_i,G) \\
D_i &= D \cup \{v_j \mid j \in [1,i-1]\}.
\end{align*}
\end{center}
This construction guarantees that the four sets are mutually disjoint and that previously visited or excluded nodes are never reconsidered. Each child node inherits the connected partial set from its parent, ensuring that connectivity is preserved at every step of the search.

\begin{example}
The {\CP} process is illustrated in Figure~\ref{fig:branching}. A parent search node $B$ generates a number of child search nodes equal to the cardinality of its reachable candidate set $C^r$. For each child search node $B_i$, the partial set $S_i$ is obtained by adding a distinct node $v_i$ from the parent’s candidate set, while the reachable candidate set $C^r_i$ is updated to include the neighbours of $v_i$ that were previously in $C^{un}$. 
Any node incorporated into a partial set at one level is simultaneously added to the exclusion set of all following sibling search nodes, ensuring that no node is reconsidered during later expansions. 
This demonstrates how {\CP} expands the search space in a connectivity-preserving manner while preventing redundant exploration.
\end{example}

The branching structure exhibits two key monotonic properties that underpin the pruning framework used in {\EBA}:
\begin{itemize}[leftmargin=*]
\item \textbf{Monotonic exclusion:} The exclusion set grows monotonically along the search tree. For sibling search nodes $B_i$ and $B_j$ where $i<j$, we have $D_i \subset D_j$, and for any parent–child relation, $D_{\text{parent}} \subseteq D_{\text{child}}$. Once a node is excluded at any stage, it remains excluded in all descendant and following sibling search nodes. This property ensures that pruning decisions are inherited by all descendants, eliminating redundant exploration.
\item \textbf{Monotonic expansion:} For every child search node, the partial set expands that of its parent, \textit{i.e.}, $S_{\text{parent}} \subseteq S_{\text{child}}$. This guarantees that the search progresses toward larger candidate subgraphs and enables branching-level pruning based on the current partial size and degree thresholds.
\end{itemize}

Branching on a search node is terminated when its partial set $S$ can no longer be expanded without violating connectivity-specifically, when the reachable candidate set $C^r$ becomes empty. However, as the number of possible connected subgraphs grows exponentially with graph size, direct enumeration quickly becomes infeasible. To address this, {\EBA} incorporates an extensive pruning framework that eliminates search nodes incapable of producing larger or feasible {\flexi} subgraphs. Throughout the search, the algorithm maintains the largest solution identified so far, $F'$, and uses degree, distance, and local structural constraints to discard unpromising regions early. This motivates the following pruning strategies.

\subsection{Pruning strategies: Degree-based pruning}\label{subsec:5_2}

As {\CP} guarantees exhaustive exploration of all connected subgraphs, efficient pruning is essential to make {\EBA} practical. This section introduces pruning strategies derived from the degree properties of {\flexi}. The main idea is to identify nodes that cannot participate in any {\flexi} larger than the current best solution $F'$, thereby allowing entire search nodes to be safely discarded. We begin by establishing a theoretical lower bound on node degree for subgraphs larger than $F'$, derived directly from the definition of {\flexi}.

\begin{lemma}\label{lem:degree_bound}
Given a graph $G=(V,E)$, a parameter $\tau \in [0,1)$, and a {\flexi} $H \subseteq V$, for any node $x \in H$, the size of $H$ is bounded as follows:
\begin{align*}
|H| \le \lfloor (d(x, G) + 1)^{1/\tau} \rfloor.
\end{align*}
\end{lemma}

\begin{proof}
From the definition of {\flexi}, the following inequality holds for any node $x \in H$:
\[
\lfloor |H|^{\tau} \rfloor \le \delta(G[H]) \le d(x, G).
\]
Rearranging gives
\[
|H| \le (d(x, G) + 1)^{1/\tau}. 
\]
\end{proof}

Lemma~\ref{lem:degree_bound} implies that nodes with insufficient degree cannot contribute to any {\flexi} larger than the current best solution $F'$. To formalise this, we define the minimum degree threshold $\theta_{F'}$ as the smallest value such that any node with degree less than $\theta_{F'}$ cannot belong to a {\flexi} subgraph larger than $|F'|$:
\begin{align*}
\theta_{F'} = \left\lceil (|F'| + 1)^{\tau} - 1 \right\rceil.
\end{align*}

With this threshold, we define the \textit{adjusted degree}, which represents the effective degree of a node within the current search node $B = (S, C^r, C^{un}, D)$. It counts the number of neighbours of a node that remain within the current search scope (i.e., not yet excluded by $D$).

\begin{definition}[Adjusted degree]
Given a graph $G = (V, E)$, a node $v \in V$, and a search node $B = (S, C^r, C^{un}, D)$, the adjusted degree of node $v$ in $B$ is defined as:
\begin{align*}
\label{eq:adjusted_degree}
d_a(v, G, B) = d(v, G[S \cup C^r \cup C^{un}]).
\end{align*}
\end{definition}

When the context is clear, we simply write $d_a(v)$. Combining $\theta_{F'}$ and the adjusted degree yields the following degree-based pruning rule.

\begin{figure}[t]
\centering
\includegraphics[width=0.6\linewidth]{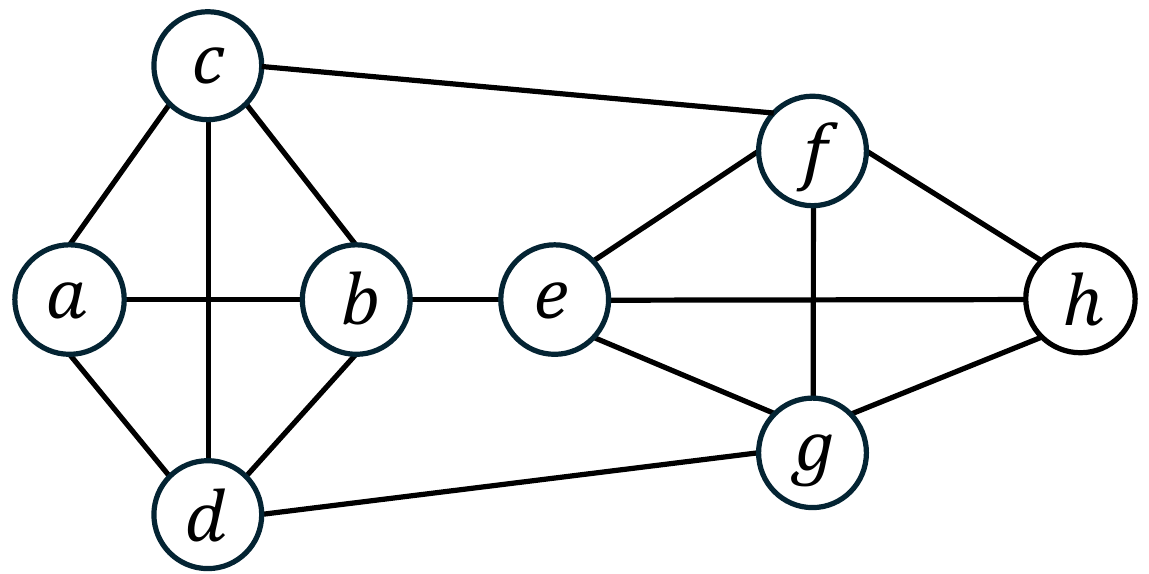}
\caption{Toy example}
\vspace{-0.7cm} 
\label{fig:toy}
\end{figure}

\begin{rules}\label{rule:rule1}
Given $F'$ the largest {\flexi} identified so far and a search node $B = (S, C^r, C^{un}, D)$ with $S \neq \emptyset$, if the minimum adjusted degree among nodes in $S$ is smaller than the threshold $\theta_{F'}$, i.e., $\min_{v \in S} d_a(v) < \theta_{F'},$ then the search node $B$ can be safely pruned, because no extension of $S$ can yield a {\flexi} larger than $F'$.
\end{rules}

Rule~\ref{rule:rule1} eliminates infeasible search nodes by verifying whether all nodes in the current partial set meet the minimum degree requirement implied by $\theta_{F'}$. This rule works synergistically with the monotonic exclusion property: once a search node is pruned by Rule~\ref{rule:rule1}, all its descendants inherit the same exclusion set and can be skipped. Moreover, sibling search nodes that share the same infeasible partial set in their parent can also be omitted without expansion.

Furthermore, an additional pruning can be obtained by deriving an upper bound on the possible {\flexi} size. 

\begin{rules}\label{rule:rule2}
For a search node $B = (S, C^r, C^{un}, D)$ with $S \neq \emptyset$, and a parameter $\tau \in [0,1)$, let 
$F_{\max} = \left\lfloor (\min_{v \in S} d_a(v) + 1)^{1/\tau} \right\rfloor$ be the maximum size of a {\flexi} achievable from $B$. If $|S| > F_{\max}$ or $|S| = F_{\max}$ and $S$ does not form a valid {\flexi}, the search node $B$ can be safely pruned.
\end{rules}

Rule~\ref{rule:rule2} terminates search nodes whose partial solution has already reached its theoretical size bound without satisfying the {\flexi} condition, thus avoiding futile expansion.

To maximise the effectiveness of the above pruning rules, {\EBA} determines the branching order of candidate nodes as follows. For a search node $B = (S, C^r, C^{un}, D)$, if $S$ is empty, the nodes in $C^{un}$ are sorted in ascending order of degree, and the algorithm expands the first node accordingly. Otherwise, if $S$ is not empty, when a node $v\in C^r$ is inserted into $S$, its neighbours from $C^{un}$ are moved to $C^r$ while maintaining the degree-ascending order of $C^r$. 

This ordering enables early pruning through two complementary effects. First, by placing low-degree nodes earlier in $C^r$, their neighbours are excluded in the initial branching steps, and these low-degree nodes are quickly identified as unpromising due to insufficient remaining neighbours to meet degree constraints. Second, deferring high-degree nodes until later in the branching order increases pruning efficiency. By the time a high-degree node is added to $S$, many of its neighbours have already been excluded due to earlier branching, leaving a smaller set of remaining candidates.

\begin{figure}[t]
\centering
\includegraphics[width=0.99\linewidth]{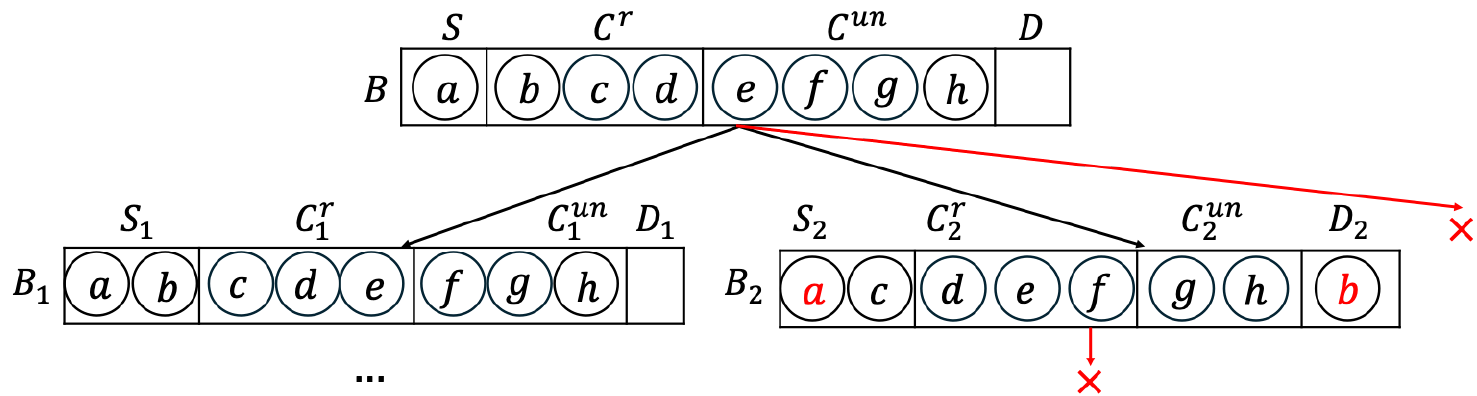}
\vspace{-0.3cm}
\caption{Degree-based pruning}
\vspace{-0.8cm} 
\label{fig:prune2}
\end{figure}

\begin{example}
Figure~\ref{fig:prune2} illustrates the degree-based pruning process on the graph in Figure~\ref{fig:toy}. Consider the case where $\tau = 0.75$ and the current best {\flexi} $F'$ has size $6$. Therefore, $\theta_{F'} = \left\lceil (6 + 1)^{0.75} - 1 \right\rceil = 3$. For the search node $B = (\{a\}, \{b, c, d\}, \{e, f, g, h\}, \emptyset)$, the partial set of its second child search node $S_2$ includes node $c$ and moves node $b$ to the exclusion set $D_2$. In this search node, the adjusted degree of node $a$ becomes $2$, which is smaller than $\theta_{F'} = 3$. Consequently, this search node is deemed infeasible and is pruned. By the monotonic exclusion property, this pruning effect propagates to all descendant, and following sibling search nodes that inherit the same infeasible partial set from their parent can also be omitted without explicit expansion. 
\end{example}

\subsection{Pruning strategies: Diameter-based pruning}

This pruning strategy removes infeasible candidate nodes within each search node by combining degree and distance constraints of the induced subgraph. The key idea is that when two nodes are far apart, satisfying the minimum degree requirement necessitates a larger number of intermediate nodes. This relationship defines a lower bound on the subgraph size required to include both nodes. When this bound exceeds the size of the current candidate scope, the corresponding node can be safely discarded since no valid {\flexi} can be formed under such conditions. The diameter of a graph $G=(V,E)$ is defined as the maximum shortest-path distance between any pair of nodes in $G$, that is, $\max\{\operatorname{sp}_G(u,v)\mid u,v\in V\}$.

\begin{lemma}\label{lem:degree_diameter}
Connected graph of minimum degree $k \ge 1$ and diameter $L \ge 1$ must contain at least $n(k,L)$ nodes, where
\begin{align*}
n(k,L) =
\begin{cases}
k + L, & \hspace{-2em}\text{if } 1 \le L \le 2 \text{ or } k = 1,\\[3pt]
k + L + 1 + \lfloor \frac{L}{3} \rfloor (k - 2), & \text{otherwise.}
\end{cases}
\end{align*}
\end{lemma}

Lemma~\ref{lem:degree_diameter} provides a lower bound on the number of nodes required to maintain both degree and diameter constraints~\cite{yao2021efficient}. This bound can be directly applied to {\EBA} by evaluating whether the current candidate is sufficiently large to contain any valid {\flexi}. If the bound exceeds the number of available nodes in the current search node, no feasible subgraph can exist, leading to the following pruning rule.

\begin{rules}\label{rule:rule3}
Let $B = (S, C^r, C^{un}, D)$ be the current search node, and let $H = S \cup C^r \cup C^{un}$ with induced subgraph $G[H]$. 
Let $F'$ denote the largest {\flexi} identified so far and $\theta_{F'}$ the corresponding minimum degree threshold. 
For any candidate node $u \in C^r \cup C^{un}$, if
\begin{align*}
n\!\left(\theta_{F'}, \max_{v \in S}\operatorname{sp}_{G[H]}(u,v)\right) > |H|,
\end{align*}
then $u$ can be safely pruned from $B$, because no connected subgraph within $H$ can simultaneously satisfy both the minimum degree $\theta_{F'}$ and the diameter constraints.
\end{rules}

This pruning rule is valid because if the number of nodes required to form a connected subgraph with minimum degree $\theta_{F'}$ and diameter at least $\operatorname{sp}_{G[H]}(u,v)$ exceeds the maximum available solution size $|H| = |S \cup C^{r} \cup C^{un}|$, then node $u$ cannot participate in any feasible {\flexi} within the current search node. 

Building upon this degree–diameter relationship, an additional pruning condition can be derived by directly incorporating the degree constraint of {\flexi}, leading to a more localised feasibility check as described next.

\begin{figure}[t]
\centering
\includegraphics[width=0.7\linewidth]{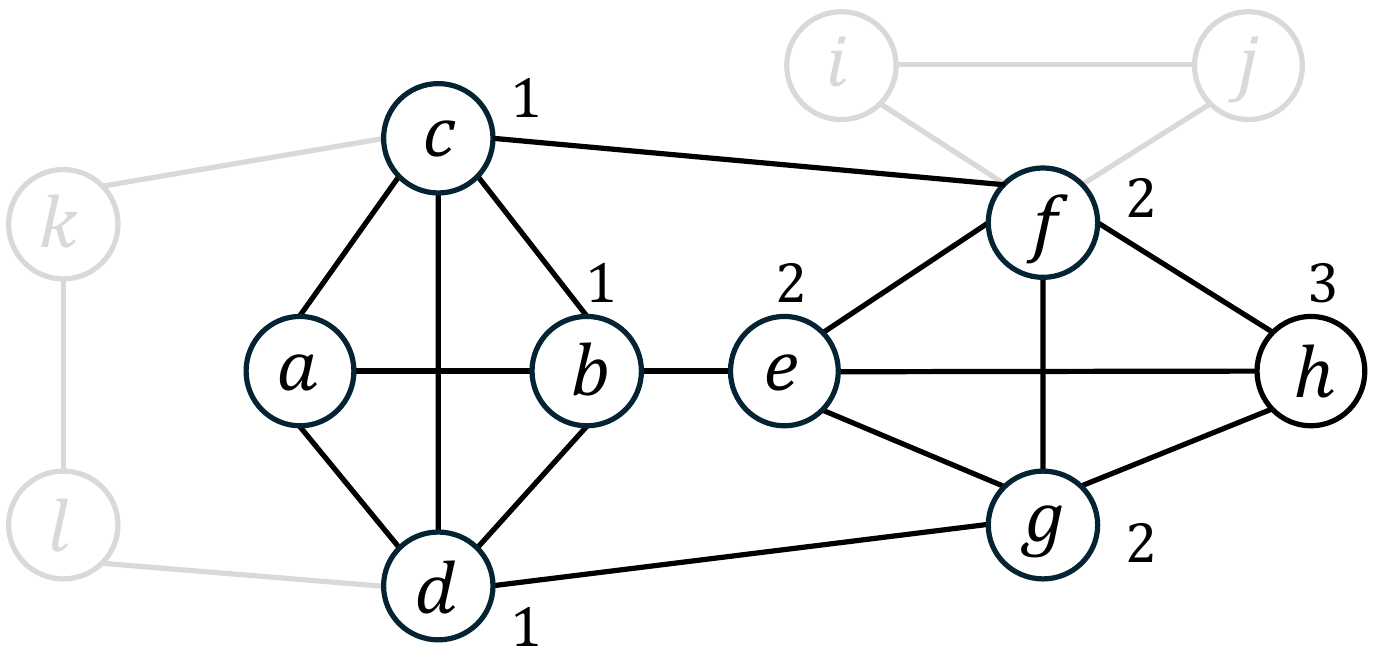}
\caption{Pruning example}
\vspace{-0.5cm} 
\label{fig:prune1}
\end{figure}

\begin{rules}\label{rule:rule4}
Let $B = (S, C^r, C^{un}, D)$ be the current search node,  parameter $\tau$, and let $H = S \cup C^r \cup C^{un}$ with induced subgraph $G[H]$. 
Given the current best {\flexi} $F'$ and $\theta_{F'}$ the corresponding minimum degree threshold, 
a candidate node $u \in C^r \cup C^{un}$ is pruned if there exists at least one node $v \in S$ such that
\begin{align*}
\left\lfloor 
n\!\left(
\theta_{F'}, \max_{x \in S} \operatorname{sp}_{G[H]}(u,x)
\right)^{\tau}
\right\rfloor > d_a(v,G).
\end{align*}
\end{rules}

This rule eliminates candidate nodes whose inclusion would require existing nodes in the partial set $S$ to exceed their maximum possible degrees, 
thereby preventing infeasible subgraph expansions.

\begin{remark}
Let $H = S \cup C^r \cup C^{un}$, with induced subgraph $G[H]$. Consider a candidate $u \in C^r \cup C^{un}$ and a node $v \in S$. Define $M$ as the smallest induced subgraph of $G[H]$ that contains $S \cup \{u\}$ and preserves the shortest-path distance between $u$ and $v$, that is, $\operatorname{sp}_M(u,v) = \operatorname{sp}_{G[H]}(u,v)$.  
If the degree required for $M$ to satisfy the {\flexi} condition exceeds the maximum possible degree of $v$, namely  
\begin{align*}
\lfloor n(\theta_{F'}, \operatorname{sp}_M(u,v))^{\tau} \rfloor > d_a(v, G),
\end{align*}
then $u$ and $v$ are mutually incompatible and cannot co-exist in any valid {\flexi} within the current search node.
\end{remark}

\spara{Incremental maintenance of distances.} To avoid all-sources recomputation at every search node, we maintain for each $u \in H= C^r \cup C^{un}$ the quantity
\begin{align*}
\text{dist}(u) = \max_{v \in S} \operatorname{sp}_{G[H]}(u,v),
\end{align*}
which represents the farthest distance from $u$ to any node in the current partial set $S$. When moving from a parent to a child search node by adding a new node $w$ to $S$, a single-source shortest path computation is executed from $w$ in $G[H]$. For each node $u \in H$, the recorded distance is updated as
\begin{align*}
\text{dist}(u) = \max\{\text{dist}(u), \operatorname{sp}_{G[H]}(w,u)\}.
\end{align*}
This incremental update preserves the distance information required for the diameter-based pruning rules while avoiding repeated all-sources computations. This directly provides the distance term required by Rule~\ref{rule:rule3} and Rule~\ref{rule:rule4}, since both rules depend on the maximum shortest-path distance from a candidate to the current partial set.

\begin{example}
Figure~\ref{fig:prune1} illustrates diameter-based pruning. The number above each node represents its shortest-path distance from node~$a$. Suppose $\tau = 0.75$ and the largest {\flexi} found so far, $F'$, has size $6$, giving the minimum degree threshold $\theta_{F'} = 3$. In the search node $B = (\{a\}, \{b, c, d\}, \{e, f, g, h\}, \emptyset)$, nodes $e$, $f$, and $g$ are each at distance $2$ from node $a$. According to Lemma~\ref{lem:degree_diameter}, a subgraph satisfying $\theta_{F'}=3$ and distance $2$ must contain at least $5$ nodes; since these nodes meet this requirement with degrees $\lfloor 5^{0.75} \rfloor = 3$, they remain as candidates. In contrast, node $h$ is at distance $3$ from node $a$, which requires a minimum subgraph size of $8$ to satisfy the degree condition. The degree of node $a$, $d(a) = 3$, cannot meet this stricter requirement $\lfloor 8^{0.75} \rfloor = 4$, so $h$ is safely pruned from the current search node.
\end{example}

\subsection{Pruning strategies: Follower-based pruning}
Building on the degree-based pruning rules, we further reduce redundant exploration by exploiting the cascading effect of node exclusion. 
When a node from the partial set $S$ is moved to the exclusion set $D$ in a sibling search node, this exclusion can lower the adjusted degrees of its neighbours. As a result, other nodes that previously satisfied the degree threshold $\theta_{F'}$ may now fall below it, becoming infeasible for any larger {\flexi}. We refer to these dependent nodes as the \textit{followers} of the excluded node.

\begin{definition}[Follower]
Given $F'$ the largest {\flexi} identified so far and its corresponding minimum degree threshold $\theta_{F'}$, let $B = (S, C^r, C^{un}, D)$ be a search node. 
When a node $v \in S$ is moved to the exclusion set $D$, any node $u \in (S \cup C^r \cup C^{un}) \setminus \{v\}$ such that
\begin{align*}
d_a(u, G, B \setminus \{v\}) < \theta_{F'}
\end{align*}
is defined as a \textit{follower} of $v$, denoted by $F_v$.
\end{definition}

By Lemma~\ref{lem:degree_bound}, every node participating in a {\flexi} larger than $F'$ must have degree at least $\theta_{F'}$. Hence, when a node's adjusted degree falls below $\theta_{F'}$ after exclusion of its neighbour, it can no longer belong to any valid {\flexi} larger than $F'$. Thus, we can derive the following pruning rule:

\begin{rules}
Let $F'$ denote the largest {\flexi} identified so far, and $\theta_{F'}$ the corresponding minimum degree threshold. For a search node $B_i = (S_i, C^r_i, C^{un}_i, D_i)$, when a node $v \in S_i$ is transferred to the exclusion set of its sibling search node $B_{i+1}$, all of its followers $F_v$ can be removed from the sibling search node and inserted into exclusion set $D_{i+1}$.
\label{rule:rule5}
\end{rules}

The rule is applied iteratively: whenever new followers are excluded, their removal may generate additional followers whose adjusted degrees are less than $\theta_{F'}$. This process continues until no further followers are produced, yielding a stable state in which every remaining node in the search node potentially satisfies the degree threshold. Thus, it can achieve substantial reductions in the size of each search node and improve overall pruning efficiency.

\subsection{Pruning strategies: Heuristic-based pruning}

While the preceding rules effectively remove numerous unpromising candidate nodes and search nodes, initiating the search with $F' = \emptyset$ across the entire graph incurs significant computational overhead, as the algorithm must explore an excessively large search space before finding any solution. To handle this inefficiency, we employ a heuristic-based pruning strategy with {\FPA}, which rapidly identifies a feasible {\flexi} to initialise $F'$. This early estimate enables the efficient removal of nodes that cannot participate in any larger {\flexi}, thereby reducing the initial search space.

By simplifying the inequality in Lemma~\ref{lem:degree_bound}, we directly derive the following condition:
\begin{align*}
&(d(v, G) + 1)^{1/\tau} < |F'| + 1 \\
&\Rightarrow d(v, G) < \lceil(|F'| + 1)^{\tau} - 1\rceil = \theta_{F'}.
\end{align*}

\begin{rules}\label{rule:rule6}
Given a graph $G = (V, E)$, a parameter $\tau \in [0,1)$, and the largest {\flexi} $F'$ identified so far, a node $v \in V$ can be safely discarded if $d(v, G) < \theta_{F'}$.
\end{rules}

Any node $v \in V$ with $d(v, G) < \theta_{F'}$ can therefore be excluded from the initial search space, and this filtering is re-applied whenever $\theta_{F'}$ is updated. In summary, $\theta_{F'}$ serves as a dynamic lower bound on node degree. It enables early elimination of infeasible nodes before exact enumeration and allows progressively stronger pruning as larger {\flexi} solutions are identified, thereby reducing the effective search space throughout the enumeration process.

\begin{example}
Figure~\ref{fig:prune1} illustrates how Rule~\ref{rule:rule6} is applied when $|F'| = 6$ and $\tau = 0.75$. To form a {\flexi} larger than $|F'|$, each node in the subgraph must have degree at least $\theta_{F'} = \lceil(6 + 1)^{0.75} - 1\rceil = 3$. Consequently, nodes with degree $2$, shaded in grey in the figure, can be safely pruned from the search space.
\end{example}

\SetAlgoNoEnd
\begin{algorithm}[t]
\footnotesize
\SetAlgoLined
\KwIn{Graph $G = (V, E)$, parameter $\tau$}
\KwOut{Maximum {\flexi} $F'$}
\SetKw{return}{return}
\SetKw{OR}{OR}
\SetKw{and}{and}
\SetKw{break}{break}
\SetKw{continue}{continue}
\SetKwData{false}{False}
\SetKwData{true}{True}
\SetKwFunction{approxDensest}{2-approx-densest-subgraph}
\SetKwFunction{pruneByDegree}{degreePrune}
\SetKwFunction{isFlexi}{isflexi}
\SetKwFunction{checkUpperBound}{checkUpperBound}
\SetKwFunction{selectNode}{selectNode}
\SetKwFunction{updateBestSolution}{updateBestSolution}
\SetKwFunction{checkFeasibility}{checkFeasibility}
\SetKwFunction{applyRules}{diaPrune}
\SetKwFunction{EBARec}{EBA-Recursive}
\SetKwFunction{addFollower}{addFollower}
\SetKwFunction{computeFollower}{computeFollower}
\SetKwFunction{checkaAdjDegree}{checkaAdjDegree}
\SetKwFunction{fpa}{FPA}
\SetKwIF{If}{ElseIf}{Else}{if}{:}{else if}{else}{}
\SetKwFor{While}{while}{:}{}

$F' \leftarrow$ \fpa{$G, \tau$}\;
$\theta_{F'} \leftarrow  \left\lceil (|F'| + 1)^\tau - 1 \right\rceil$ \;
$V' \leftarrow$ \pruneByDegree{$G, \theta_{F'}$} \tcp*{Rule~\ref{rule:rule6}}
\EBARec{$\emptyset, \emptyset, V', \emptyset$}\;
\return $F'$\;

\BlankLine
\SetKwProg{Proc}{Procedure}{}{}
\Proc{\EBARec{$S, C^r, C^{un}, D$}}{
    
    \If{$\min_{v \in S} d_a(v) < \theta_F'$}{ 
        \return \tcp*{Rule~\ref{rule:rule1}}
    }

    \If{\isFlexi{$S,\tau$}=$\false$ \and $|S| \geq \lfloor (\min\limits_{v \in S} d_a(v) + 1)^{1/\tau} \rfloor$}{ 
    \return \tcp*{Rule~\ref{rule:rule2}}
    }

    $Fol \leftarrow \emptyset$\;
    Create search nodes $\{B_1, B_2, \ldots, B_{|C^r|}\}$ for each $v \in C^r$\;
    \For{each branch $B_i$}{
        $D_i \leftarrow D_i \cup Fol$\; 
        \If{$|S_i| > |F'|$ \and \isFlexi{$S_i, \tau$}}{
            $F' \leftarrow S_i$\;
            $\theta_F' \leftarrow  \left\lceil (|F'| + 1)^\tau - 1 \right\rceil$ \;
            
            $V' \leftarrow$ \pruneByDegree{$G, \theta_{F'}$} \tcp*{Rule~\ref{rule:rule6}}
        }

        $Fol \leftarrow$ \computeFollower{$G, B_i, v_i$} \tcp*{Rule~\ref{rule:rule5}}
        
        $C^r_i \leftarrow C^r \cup N_{C^{un}}(v_i) \setminus (S_i\cup D_i)$\;
        $C^{un}_i \leftarrow C^{un} \setminus (S_i\cup C^r_i \cup D_i)$\;
        
        $C^r, C^{un}_i \leftarrow$ \applyRules{$S_i,C^r_i,C^{un}_i,\theta_{F'}$} \tcp*{Rule~\ref{rule:rule3}\&\ref{rule:rule4}}

        \If{\checkaAdjDegree{$\theta_{F'}$,$S_i, C^r_i, C^{un}_i$} = \false}{
            \continue \tcp*{Rule~\ref{rule:rule1}}
        }
        
        \EBARec{$S_i, C^r_i, C^{un}_i, D_i$}\;
    }
}
\caption{\mbox{\FuncSty{{\EBA}}: Efficient branch and bound algorithm}}
\label{alg:EBA}
\end{algorithm}
\vspace{-1.0em}

\subsection{Algorithmic procedure}

The overall workflow of {\EBA} is summarised in Algorithm~\ref{alg:EBA}, which systematically integrates {\CP} with all proposed pruning strategies. The algorithm first computes a feasible {\flexi} to initialise the lower bound $F'$ using the heuristic {\FPA}, which provides an initial degree threshold $\theta_{F'}$. Subsequently, Rule~\ref{rule:rule6} (heuristic-based pruning) is applied to immediately exclude nodes that do not satisfy the minimum degree threshold, resulting in a reduced search space. The search then proceeds from the root search node $(\emptyset, \emptyset, V', \emptyset)$, recursively exploring the search tree.

At each recursive step, a search node $(S, C^r, C^{un}, D)$ is pruned if (i) the minimum adjusted degree in $S$ fails to satisfy $\theta_{F'}$ (Rule~\ref{rule:rule1}), (ii) the partial set $S$ has already reached its possible maximum size while remaining invalid (Rule~\ref{rule:rule2}), or (iii) the node cannot be expanded further. Otherwise, the search node generates child nodes according to Section~\ref{sec:subsec5_1}, each corresponding to the inclusion of a new candidate.  

For each child $B_i$, followers of preceding sibling search nodes are excluded by Rule~\ref{rule:rule5}, and if the partial solution $S_i$ exceeds the current best solution and satisfies the {\flexi} condition, the algorithm updates the best-so-far solution and reapplies Rule~\ref{rule:rule6} to remove nodes that have become unpromising. In the next step, followers of the node added in the current partial set are computed, and diameter-based rules (Rules~\ref{rule:rule3} and~\ref{rule:rule4}) filter out nodes that are too distant to form a valid {\flexi}. Search nodes that pass these checks are explored recursively.

%% file: 06_experiments.tex
\section{Experiments}\label{sec:experiments}

In this section, we present a comprehensive experimental evaluation to demonstrate the effectiveness and efficiency of the proposed algorithms. Specifically, we aim to answer the following evaluation questions (EQs):

\begin{itemize}[leftmargin=*] 
    \item \textbf{EQ1. Effectiveness:} How good is the solution quality of the new heuristic algorithm ({\FPA})? 
    \item \textbf{EQ2. Efficiency:} How do the proposed methods perform in terms of computational time in diverse real-world networks? 
    \item \textbf{EQ3. Effect of $\tau$:} How does the parameter $\tau$ influence running time performance and solution size? 
    \item \textbf{EQ4. Synthetic analysis:} How well do our algorithms scale with increasing graph size and varying average degree? 
    \item \textbf{EQ5. Ablation study:} What is the individual contribution of each pruning rule to overall algorithm performance? 
    \item \textbf{EQ6. Case study:} How effectively can {\flexi} identify meaningful communities in real-world networks?
\end{itemize}

\subsection{Experiment setting}\label{subsec:6_1}

\spara{Real-world datasets.} 
We use 9 real-world networks from KONECT~\cite{konect}, covering various domains such as social, collaboration, and communication networks. Table~\ref{tab:data} summarises their key statistics. ``Avg.~Deg'' denotes the average degree, and ``CC'' represents the global clustering coefficient.

\spara{Algorithms for comparison.} To evaluate the effectiveness and efficiency of our algorithms comprehensively, we select baseline algorithms representative of state-of-the-art approaches:
\begin{itemize}[leftmargin=*]
    \item \textbf{Baseline algorithm}: {\NPA}~\cite{flexi}, an existing heuristic algorithm for {\flexi} search.
    \item \textbf{Maximum cohesive subgraph algorithms}: Maximum clique~\cite{bron1973algorithm,tomita2006worst,cazals2008note}, $k$-plex~\cite{chang2022efficient}, and quasi-clique~\cite{yu2023fast} algorithms used to compare with alternative cohesive subgraph models.
\end{itemize}

\spara{Parameter setting.} Unless stated otherwise, the parameter $\tau$ is fixed at $0.9$.
While previous studies on quasi-clique and related models~\cite{yu2023fast,pei2005mining,khalil2022parallel,guo2020scalable} report that density thresholds between $0.5$ and $0.7$ yield interpretable communities, the exponent $\tau$ in {\flexi} controls cohesion sub-linearly with size. Empirically, $\tau=0.9$ produces comparable effective density levels in medium-sized subgraphs (for instance, a 50-node subgraph requires a minimum degree of $\lfloor 50^{0.9}\rfloor = 34$, or about $68\%$ connectivity), which corresponds to the typical density range used in prior models. As shown in EQ3 (Figure~\ref{fig:tau}), the impact of $\tau$ varies across datasets-some networks exhibit slightly increasing running time or smaller subgraphs as $\tau$ grows, while others remain relatively stable. We adopt $\tau=0.9$ as it offers a balanced trade-off between cohesiveness and computational efficiency across all datasets.

\spara{Experimental environment.} All algorithms were implemented in Python using the NetworkX library~\cite{hagberg2008exploring}. Experiments were performed on a machine equipped with an Intel Xeon 6248R CPU and 256GB RAM running Ubuntu 20.04.

\spara{Code availability.}
All source codes for our results are available at 
\underline{\url{https://github.com/Song1940/MaximumFlexiClique}}.

\begin{table}[t]
\centering
\caption{Real‐world datasets}
\label{tab:data}
\centering
\small
\begin{tabular}{c||r|r|r|r}
\hline
\textbf{Dataset} & $\boldsymbol{|V|}$ & $\boldsymbol{|E|}$ & \textbf{Avg. Deg} & \textbf{CC} \\ 
\hline\hline
\textbf{Karate}  &       $34$ &       $78$ & $4.59$ & $0.255$ \\ \hline
\textbf{Polbooks}&      $105$ &      $441$ & $8.40$ & $0.348$ \\ \hline
\textbf{Football}&      $115$ &      $613$ & $10.66$ & $0.407$ \\ \hline
\textbf{Polblogs}&    $1,224$ &   $16,715$ & $27.31$ & $0.226$ \\ \hline
\textbf{Erdös}&    $6,927$ &   $11,850$ & $3.42$ & $0.035$ \\ \hline
\textbf{PGP}     &   $10,680$ &   $24,316$ & $4.55$ & $0.378$ \\ \hline
\textbf{Amazon}  &  $334,863$ &  $925,872$ & $5.53$ & $0.205$ \\ \hline
\textbf{DBLP}    &  $317,080$ &$1,049,866$ & $6.62$ & $0.306$ \\ \hline
\textbf{Florida} &$1,070,376$ &$1,343,951$ & $2.51$ & $0.026$ \\ 
\hline 
\end{tabular}
\vspace{-0.5cm}
\end{table}

\subsection{Experiment results}

\begin{figure}[h]
\centering
\vspace{-0.3cm}
\includegraphics[width=0.95\linewidth]{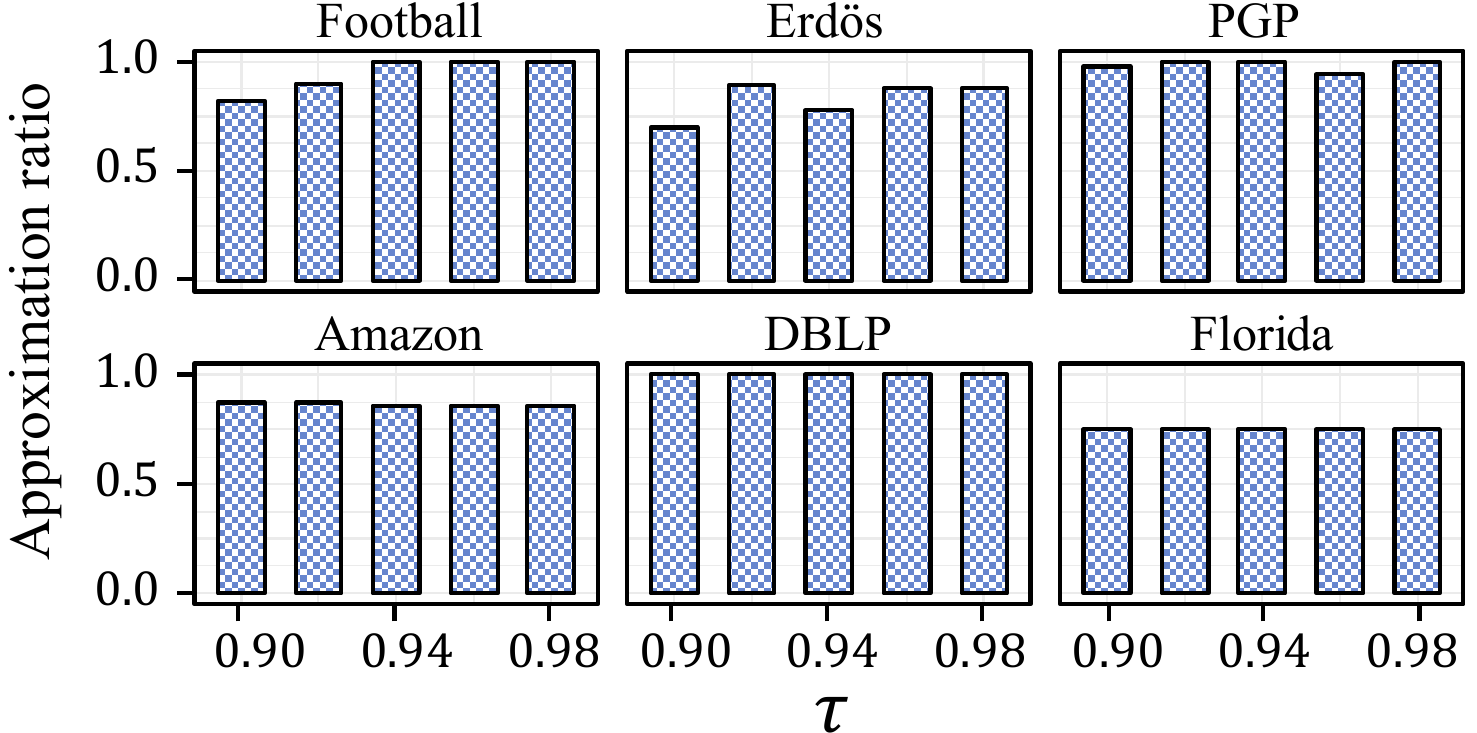}
\caption{EQ1. Effectiveness}
\vspace{-0.4cm} 
\label{fig:approx}
\end{figure}

\spara{EQ1. Effectiveness.} 
To evaluate the effectiveness of our heuristic algorithm {\FPA}, we compare its solution quality with exact results across a range of real-world networks under different values of the parameter $\tau$. As shown in Figure~\ref{fig:approx}, {\FPA} consistently yields high-quality solutions. Notably, it finds the optimal solutions on DBLP for all $\tau$ values, and attains approximation ratios above $90\%$ on most datasets. Even for more challenging networks, such as Erdös and Florida, {\FPA} achieves average ratios of $82\%$ and $75\%$ of the optimum, respectively. 
These results confirm that {\FPA} provides high-quality initial solutions, which serve as a strong starting point for the subsequent refinement phase of {\EBA}. Overall, the heuristic achieves near-optimal effectiveness with significantly reduced computational cost compared to exact methods.

\begin{table}[h]
\centering
\caption{EQ2. Efficiency (running time: sec.)}
\label{tab:eq1}
\small
\begin{tabular}{c||r|r|r|r}
\hline
\textbf{Dataset} 
& \textbf{\NPA}~\cite{flexi} 
& \textbf{\FPA} 
& \textbf{\FastQC}~\cite{yu2023fast}
& \textbf{\EBA} \\
\hline\hline
\textbf{Karate}   
& 0.001   & 0.001   & 0.002  & 0.004   \\ \hline
\textbf{Polbooks} 
& 0.011   & 0.028   & 0.021  & 0.317   \\ \hline
\textbf{Football} 
& 0.024   & 0.023   & 0.029  & 0.193   \\ \hline
\textbf{Polblogs} 
& 1.804   & 0.025   & 0.776  & 0.084   \\ \hline
\textbf{Erdös}    
& 0.176   & 0.066   & 0.165  & 0.769   \\ \hline
\textbf{PGP}      
& 0.955   & 0.093   & 0.487  & 0.192   \\ \hline
\textbf{Amazon}   
& 128.816 & 7.024   & 11.501 & 16.458  \\ \hline
\textbf{DBLP}     
& 69.629  & 10.574  & 13.923 & 19.865  \\ \hline
\textbf{Florida}  
& 14.931  & 7.953   & 13.027 & 23.650  \\
\hline
\end{tabular}
\end{table}

\spara{EQ2. Efficiency.} 
Table~\ref{tab:eq1} reports the running time of the proposed algorithms, {\FPA} and {\EBA}, together with the baseline {\NPA}~\cite{flexi}, across a range of real-world networks. Across most of the datasets, {\FPA} consistently outperforms {\NPA}. This improvement primarily arises from the use of an efficient data structure for maintaining articulation points, which enables incremental connectivity updates. In contrast, {\NPA} repeatedly rescans the entire candidate subgraph at each peeling step, incurring substantial overhead as the subgraph grows. Although {\EBA} introduces additional cost due to branch-and-bound exploration and iterative refinement of the initial solution produced by {\FPA}, it benefits from a strong starting bound and effective pruning. As a result, {\EBA} converges quickly on most datasets and remains practical even on graphs with tens of thousands of nodes.

We further compare {\EBA} with {\FastQC}~\cite{yu2023fast}, a representative algorithm for degree-based quasi-clique detection with $\gamma=0.9$ and minimum size threshold $\theta=1$. For quasi-cliques with $\gamma>0.5$, the induced subgraph is known to have diameter at most two~\cite{yu2023fast}, which allows {\FastQC} to restrict its search to a local neighbourhood. In contrast, the {\flexi} model is defined by a sub-linear degree constraint, under which no bounded-diameter guarantee can be derived, requiring {\EBA} to explore a broader search space. Despite this structural difference, {\EBA} demonstrates running times comparable to {\FastQC} on most datasets, highlighting the efficiency of the proposed algorithmic design and pruning strategies.

\begin{figure}[h]
\includegraphics[width=0.99\linewidth]
{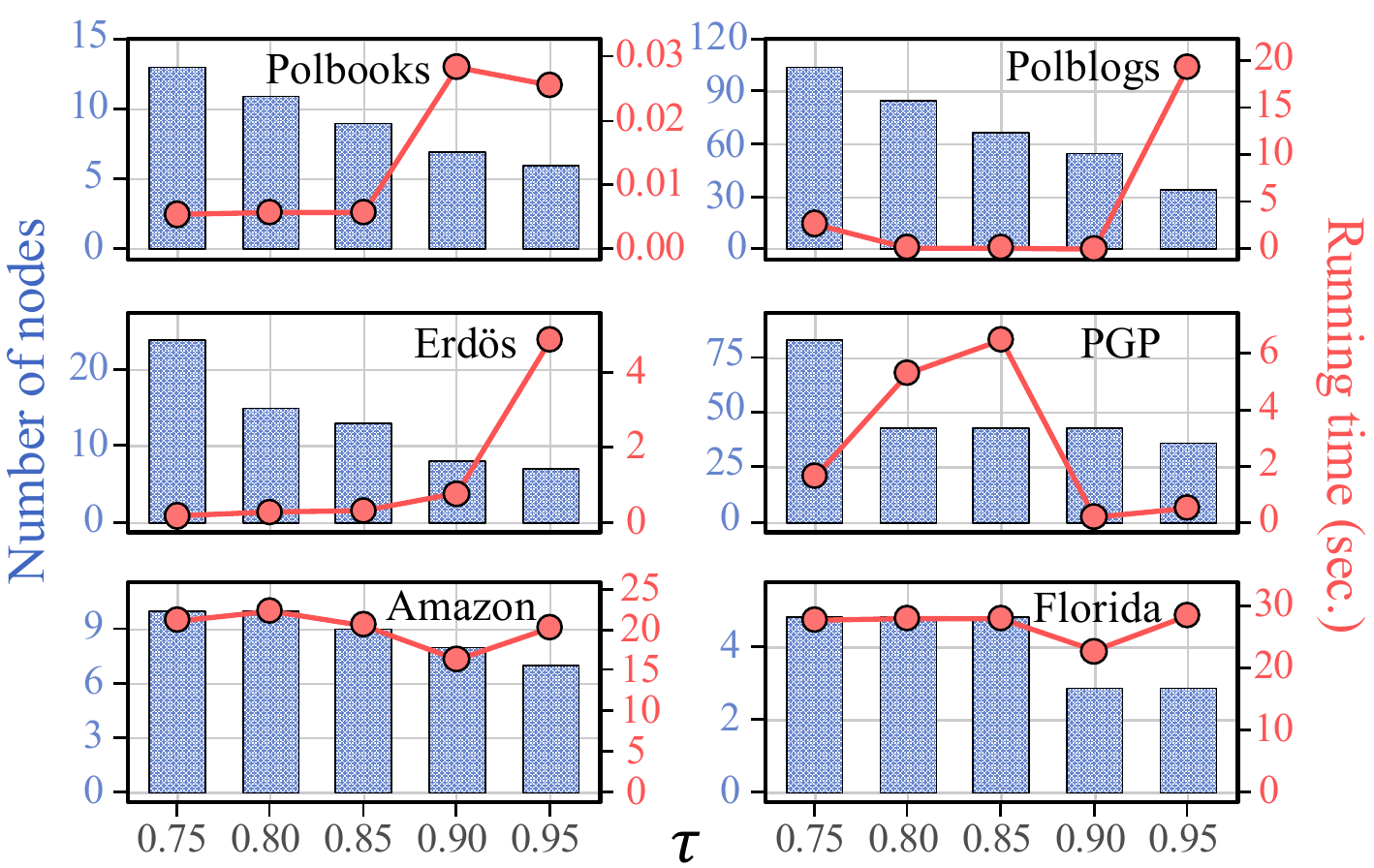}
\caption{EQ3. Effect of $\tau$}
\vspace{-0.5cm} 
\label{fig:tau}
\end{figure}

\spara{EQ3. Effect of $\tau$.} Figure~\ref{fig:tau} presents the running time of {\EBA} and the corresponding {\flexi} size as $\tau$ ranges from $0.75$ to $0.95$. There is no consistent trend in running time across networks with increasing $\tau$: for example, the running time consistently grows with $\tau$ on the Erdös dataset, whereas no such trend is observed in other datasets. This variation can be attributed to structural differences among networks, such as sparsity and degree distribution, which affect pruning efficiency and the expansion of the search tree. Since {\EBA} starts from the reduced searching space by {\FPA}'s solution, $\tau$ also indirectly affects running time through changes in the quality and size of that starting point. Nevertheless, $\tau$ has a clear and consistent effect on the solution size across all datasets. As $\tau$ increases, the degree constraint becomes stricter, resulting in smaller yet denser {\flexi} subgraphs.

\begin{figure}[h]
\centering
\vspace{-0.4cm}
    \begin{subfigure}{.48\linewidth}
    \includegraphics[width=0.99\linewidth]{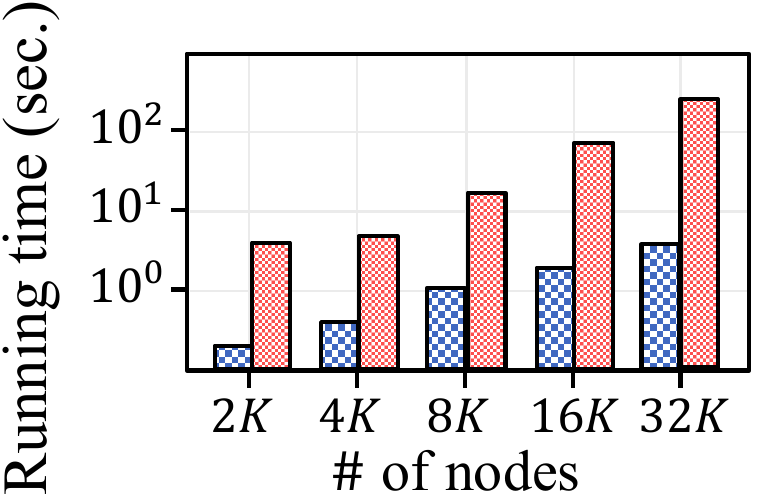}
    \captionsetup{margin={0.8cm,0cm}}
    \caption{Scalability test}
    \label{fig:syn1}
    \end{subfigure}
    \begin{subfigure}{.48\linewidth}
    \includegraphics[width=0.99\linewidth]{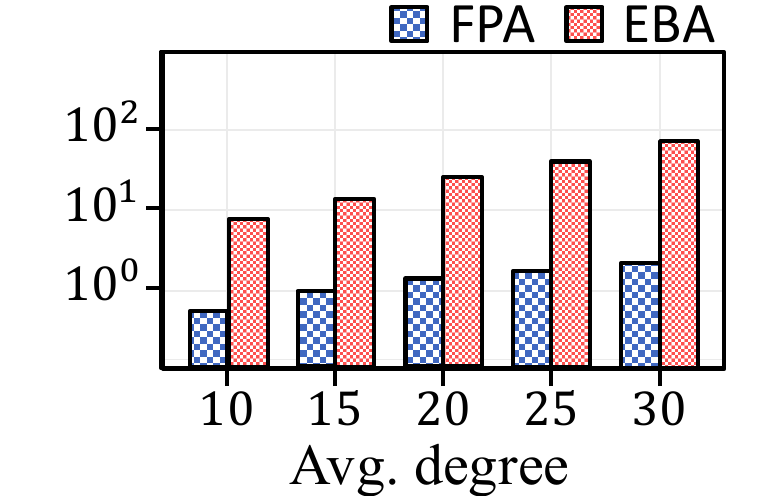}
    \caption{Effect of average degree}
    \label{fig:syb2}
    \end{subfigure}
\caption{EQ4. Synthetic analysis} 
\label{fig:synthetic}
\vspace{-0.3cm}
\end{figure}

\spara{EQ4. Synthetic analysis.} 
To evaluate scalability, we conduct experiments on LFR synthetic networks~\cite{lancichinetti2008benchmark}. 
The synthetic networks are generated with the following configuration setting: network size $10,000$, average degree $20$, maximum degree $200$, and mixing parameter $\mu=0.1$, with other parameters set to default values.
We analyse two aspects: (1) scalability with respect to network size, varying it from $100$ to $1M$ while keeping other parameters fixed; and (2) the effect of average degree, varying it from $10$ to $30$ with the network size fixed at $10,000$. All experiments are conducted with $\tau = 0.9$. Figure~\ref{fig:synthetic} reports the results. {\FPA} exhibits near-linear scalability, finishing all the tasks within two seconds. {\EBA} requires longer computation time than {\FPA} due to its exhaustive search, but it remains efficient in practice: running time increases near-linearly with network size and average degree without explosive growth.

\begin{figure*}[t]  
\centering
\vspace{-0.3cm}
\begin{subfigure}[b]{\textwidth}  
    \centering
    \includegraphics[width=0.99\textwidth]{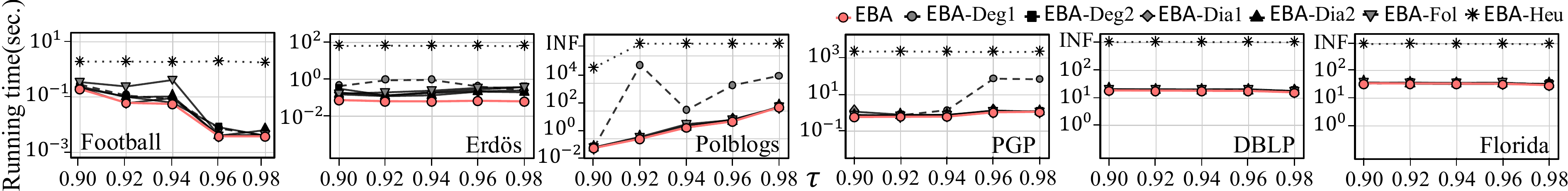}
    \caption{Impact of individual pruning rules on running time}
    \label{fig:ab_runtime}
\end{subfigure}

\vspace{-0.1cm}  

\begin{subfigure}[b]{\textwidth}  %
    \centering
    \includegraphics[width=0.99\textwidth]{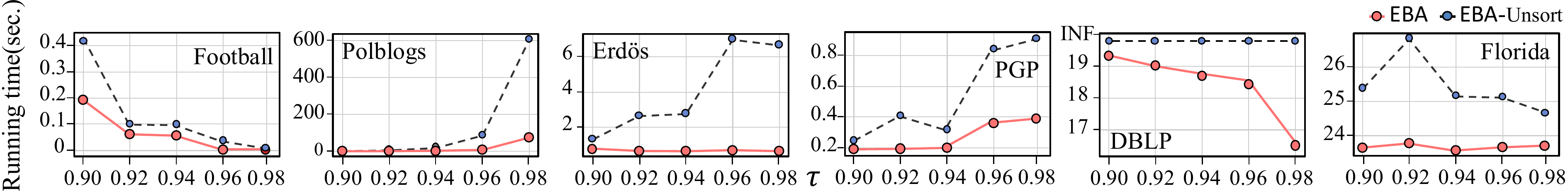}
    \caption{Effect of candidate sorting strategy}
    \label{fig:ab_sort}
\end{subfigure}
\vspace{-0.6cm}
\caption{EQ5. Ablation study}
\label{fig:ablation_study}
\vspace{-0.6cm}
\end{figure*}

\spara{EQ5. Ablation study.}
To evaluate the contribution of each pruning rule, we conduct an ablation study where each rule is removed from {\EBA} in turn: 
Rule~\ref{rule:rule1} ({\EBA}-Deg1), Rule~\ref{rule:rule2} ({\EBA}-Deg2), Rule~\ref{rule:rule3} ({\EBA}-Dia1), 
Rule~\ref{rule:rule4} ({\EBA}-Dia2), Rule~\ref{rule:rule5} ({\EBA}-Fol), and Rule~\ref{rule:rule6} ({\EBA}-Heu). 
Experiments are conducted with $\tau$ varying from $0.9$ to $0.98$ in increments of $0.02$, and cases exceeding 48 hours are marked as \text{INF}.

Figure~\ref{fig:ablation_study} shows the results on six datasets. 
{\EBA} achieves the fastest running time overall, but the effect of each rule varies by dataset. 
The heuristic-based initial search space reduction (Rule~\ref{rule:rule6}) has the most pronounced impact: without it, large datasets such as DBLP and Florida fail to terminate, and even Polblogs cannot complete execution. 
Degree-based pruning (Rules~\ref{rule:rule1}, \ref{rule:rule2}) and follower-based pruning (Rule~\ref{rule:rule5}) provide greater benefit than diameter-based rules (Rules~\ref{rule:rule3}, \ref{rule:rule4}), as they can directly terminate unpromising branches rather than merely reducing candidate sets.

We also compare the performance of maintaining the candidate set in sorted order, as proposed in Section~\ref{subsec:5_2}, against an approach that branches based on node IDs ({\EBA}-Unsorted). As shown in the Figure~\ref{fig:ab_sort}, maintaining sorted candidates yields substantial performance improvements across all cases, with an average speedup of approximately $4\times$ across all datasets. Notably, on the DBLP dataset, the unsorted approach generates numerous infeasible branches early in the search, preventing termination. This demonstrates that our degree-based candidate ordering plays a critical role in algorithm performance enhancement.

\begin{figure}[h]
\vspace{-0.2cm}
\centering
\includegraphics[width=0.9\linewidth]{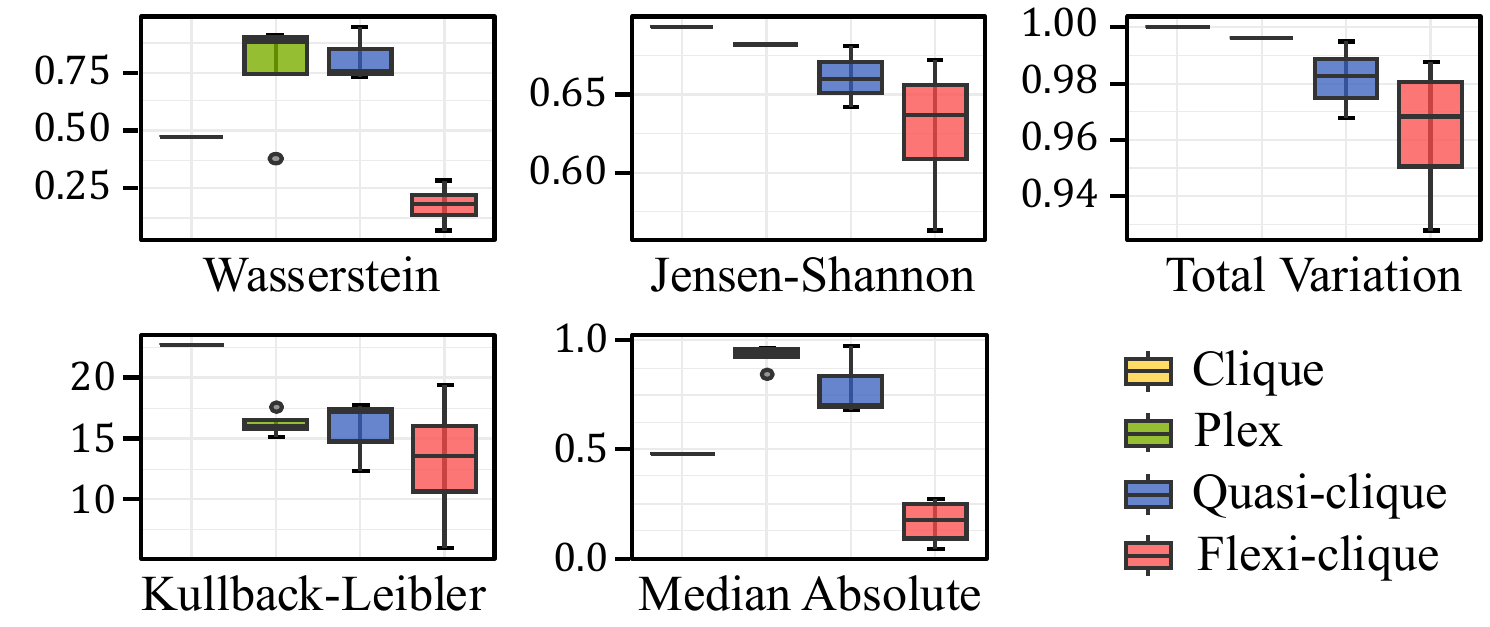}
\captionsetup{justification=centering} 
\caption{EQ6. Case study: Density distribution distance comparison among cohesive subgraph models}
\vspace{-0.4cm}
\label{fig:density_results}
\end{figure}

\spara{EQ6. Case study.}
To evaluate how effectively {\flexi} captures the structural evolution of real-world communities, we analyse the Eu-mail network~\cite{paranjape2017motifs}, a temporal graph with ground-truth community labels. 
We select a representative community and track its growth over time, representing the community on each day $i$ as $G_i=(V_i,E_i)$. 
For each $G_i$, we record the number of active nodes $|V_i|$ and edges $|E_i|$, and extract the maximum clique, $k$-plex, $\gamma$-quasi-clique, and {\flexi} subgraphs under varying parameter settings: $k$-plex with $k \in \{3,5,7\}$, $\gamma$-quasi-clique with $\gamma \in \{0.6,0.7,0.8,0.9\}$, and {\flexi} with $\tau \in \{0.6,0.7,0.8,0.9\}$.
We then compute the density $\frac{|E_i|}{|V_i|(|V_i|-1)/2}$ at each timestamp, producing temporal density distributions that describe how each model reflects community evolution.

To quantitatively compare these distributions with those of the ground-truth communities, we employ five distance metrics: Wasserstein distance~\cite{kantorovich1942translocation}, Jensen–Shannon divergence~\cite{lin2002divergence}, Total Variation distance~\cite{kantorovich1942translocation}, symmetric Kullback–Leibler divergence~\cite{kullback1951information}, and median absolute difference~\cite{conover1999practical}. 
For all metrics, smaller values indicate greater similarity to the true density distribution.

Figure~\ref{fig:density_results} summarises the results. Both clique and $k$-plex models perform poorly across all metrics due to their overly strict connectivity constraints. Cliques maintain a constant density of exactly $1.0$ by definition, while $k$-plex structures exhibit densities that converge to $1.0$ over time, both failing to capture the gradual density decline as real communities expand.
The $\gamma$-quasi-clique model provides some relaxation but still imposes a fixed high threshold, thus missing realistic temporal variation. In contrast, {\flexi} closely follows the empirical density trajectory, where density decreases naturally with community growth. Across all metrics, {\flexi} achieves the closest match to the ground truth, demonstrating that its size-adaptive constraint effectively models the dynamic evolution of real-world community structure.

%% file: 08_conclusion.tex
\section{CONCLUSION}\label{sec:conclusion}

We presented a comprehensive study of the {\flexi} model for cohesive subgraph discovery. We proved the maximum {\flexi} problem to be NP-hard and proposed two complementary algorithms: an exact branch-and-bound method ({\EBA}) and a fast heuristic ({\FPA}). Experiments showed that {\FPA} achieves near-optimal quality with large efficiency gains, while {\EBA} finds optimal solutions for moderate graphs. The {\flexi} model effectively captures size-adaptive cohesiveness, bridging cohesive connectivity models and real-world network structures. Future work will extend {\flexi} to dynamic and attributed graphs, and investigate the enumeration of maximal {\flexi} structures, which poses additional challenges due to the non-hereditary and connectivity-preserving nature of the model, and requires new algorithmic principles beyond those used for classical cohesive subgraph enumeration.



%% file: name.bib
@article{wuchty2007increasing,
  title={The increasing dominance of teams in production of knowledge},
  author={Wuchty, Stefan and Jones, Benjamin F and Uzzi, Brian},
  journal={Science},
  volume={316},
  number={5827},
  pages={1036--1039},
  year={2007},
  publisher={American Association for the Advancement of Science}
}

@inproceedings{sleator1981data,
  title={A data structure for dynamic trees},
  author={Sleator, Daniel D and Tarjan, Robert Endre},
  booktitle={Proceedings of the thirteenth annual ACM symposium on Theory of computing},
  pages={114--122},
  year={1981}
}

@article{malliaros2020core,
  title={The core decomposition of networks: Theory, algorithms and applications},
  author={Malliaros, Fragkiskos D and Giatsidis, Christos and Papadopoulos, Apostolos N and Vazirgiannis, Michalis},
  journal={The VLDB Journal},
  volume={29},
  number={1},
  pages={61--92},
  year={2020},
  publisher={Springer}
}

@article{pattillo2013maximum,
  title={On the maximum quasi-clique problem},
  author={Pattillo, Jeffrey and Veremyev, Alexander and Butenko, Sergiy and Boginski, Vladimir},
  journal={Discrete Applied Mathematics},
  volume={161},
  number={1-2},
  pages={244--257},
  year={2013},
  publisher={Elsevier}
}

@article{seidman1983network,
  title={Network structure and minimum degree},
  author={Seidman, Stephen B},
  journal={Social networks},
  volume={5},
  number={3},
  pages={269--287},
  year={1983},
  publisher={Elsevier}
}

@inproceedings{wang2022listing,
  title={Listing maximal k-plexes in large real-world graphs},
  author={Wang, Zhengren and Zhou, Yi and Xiao, Mingyu and Khoussainov, Bakhadyr},
  booktitle={Proceedings of the ACM Web Conference 2022},
  pages={1517--1527},
  year={2022}
}

@inproceedings{conte2017fast,
  title={Fast enumeration of large k-plexes},
  author={Conte, Alessio and Firmani, Donatella and Mordente, Caterina and Patrignani, Maurizio and Torlone, Riccardo},
  booktitle={Proceedings of the 23rd ACM SIGKDD international conference on knowledge discovery and data mining},
  pages={115--124},
  year={2017}
}

@inproceedings{lin2021constant,
  title={Constant approximating k-clique is w [1]-hard},
  author={Lin, Bingkai},
  booktitle={Proceedings of the 53rd Annual ACM SIGACT Symposium on Theory of Computing},
  pages={1749--1756},
  year={2021}
}

@article{batagelj2003m,
  title={An o (m) algorithm for cores decomposition of networks},
  author={Batagelj, Vladimir and Zaversnik, Matjaz},
  journal={arXiv preprint cs/0310049},
  year={2003}
}

@article{lanciano2024survey,
  title={A survey on the densest subgraph problem and its variants},
  author={Lanciano, Tommaso and Miyauchi, Atsushi and Fazzone, Adriano and Bonchi, Francesco},
  journal={ACM Computing Surveys},
  volume={56},
  number={8},
  pages={1--40},
  year={2024},
  publisher={ACM New York, NY}
}

@inproceedings{tsourakakis2013denser,
  title={Denser than the densest subgraph: extracting optimal quasi-cliques with quality guarantees},
  author={Tsourakakis, Charalampos and Bonchi, Francesco and Gionis, Aristides and Gullo, Francesco and Tsiarli, Maria},
  booktitle={Proceedings of the 19th ACM SIGKDD international conference on Knowledge discovery and data mining},
  pages={104--112},
  year={2013}
}

@article{sanei2021mining,
  title={Mining largest maximal quasi-cliques},
  author={Sanei-Mehri, Seyed-Vahid and Das, Apurba and Hashemi, Hooman and Tirthapura, Srikanta},
  journal={ACM Transactions on Knowledge Discovery from Data (TKDD)},
  volume={15},
  number={5},
  pages={1--21},
  year={2021},
  publisher={ACM New York, NY, USA}
}

@inproceedings{charikar2000greedy,
  title={Greedy approximation algorithms for finding dense components in a graph},
  author={Charikar, Moses},
  booktitle={International workshop on approximation algorithms for combinatorial optimization},
  pages={84--95},
  year={2000},
  organization={Springer}
}

@inproceedings{flexi,
  title={Flexi-clique: Exploring Flexible and Sub-linear Clique Structures},
  author={Kim, Song and Kim, Junghoon and Yoon, Susik and Kim, Jungeun},
  booktitle={Proceedings of the 33rd ACM International Conference on Information and Knowledge Management},
  pages={3832--3836},
  year={2024}
}

@inproceedings{chang2019efficient,
  title={Efficient maximum clique computation over large sparse graphs},
  author={Chang, Lijun},
  booktitle={Proceedings of the 25th ACM SIGKDD International Conference on Knowledge Discovery \& Data Mining},
  pages={529--538},
  year={2019}
}

@book{leskovec2020mining,
  title={Mining of massive data sets},
  author={Leskovec, Jure and Rajaraman, Anand and Ullman, Jeffrey David},
  year={2020},
  publisher={Cambridge university press}
}

@article{goldberg1984finding,
  title={Finding a maximum density subgraph},
  author={Goldberg, Andrew V},
  year={1984},
  publisher={University of California Berkeley}
}

@article{khalil2022parallel,
  title={Parallel mining of large maximal quasi-cliques},
  author={Khalil, Jalal and Yan, Da and Guo, Guimu and Yuan, Lyuheng},
  journal={The VLDB Journal},
  volume={31},
  number={4},
  pages={649--674},
  year={2022},
  publisher={Springer}
}

@article{ribeiro2019exact,
  title={An exact algorithm for the maximum quasi-clique problem},
  author={Ribeiro, Celso C and Riveaux, Jos{\'e} A},
  journal={International Transactions in Operational Research},
  volume={26},
  number={6},
  pages={2199--2229},
  year={2019},
  publisher={Wiley Online Library}
}

@book{karp2010reducibility,
  title={Reducibility among combinatorial problems},
  author={Karp, Richard M},
  year={2010},
  publisher={Springer}
}

@article{yao2021efficient,
  title={Efficient size-bounded community search over large networks},
  author={Yao, Kai and Chang, Lijun},
  journal={Proceedings of the VLDB Endowment},
  volume={14},
  number={8},
  pages={1441--1453},
  year={2021},
  publisher={VLDB Endowment}
}

@article{metcalfe2013metcalfe,
  title={Metcalfe's law after 40 years of ethernet},
  author={Metcalfe, Bob},
  journal={Computer},
  volume={46},
  number={12},
  pages={26--31},
  year={2013},
  publisher={IEEE}
}

@article{tverskoi2021dynamics,
  title={The dynamics of cooperation, power, and inequality in a group-structured society},
  author={Tverskoi, Denis and Senthilnathan, Athmanathan and Gavrilets, Sergey},
  journal={Scientific reports},
  volume={11},
  number={1},
  pages={18670},
  year={2021},
  publisher={Nature Publishing Group UK London}
}

@article{balasundaram2011clique,
  title={Clique relaxations in social network analysis: The maximum k-plex problem},
  author={Balasundaram, Balabhaskar and Butenko, Sergiy and Hicks, Illya V},
  journal={Operations Research},
  volume={59},
  number={1},
  pages={133--142},
  year={2011},
  publisher={INFORMS}
}

@article{guo2020scalable,
  title={Scalable mining of maximal quasi-cliques: an algorithm-system codesign approach},
  author={Guo, Guimu and Yan, Da and {\"O}zsu, M Tamer and Jiang, Zhe and Khalil, Jalal},
  journal={Proceedings of the VLDB Endowment},
  volume={14},
  number={4},
  pages={573--585},
  year={2020},
  publisher={VLDB Endowment}
}

@article{chang2022efficient,
  title={Efficient maximum k-plex computation over large sparse graphs},
  author={Chang, Lijun and Xu, Mouyi and Strash, Darren},
  journal={Proceedings of the VLDB Endowment},
  volume={16},
  number={2},
  pages={127--139},
  year={2022},
  publisher={VLDB Endowment}
}

@article{seidman1978graph,
  title={A graph-theoretic generalization of the clique concept},
  author={Seidman, Stephen B and Foster, Brian L},
  journal={Journal of Mathematical sociology},
  volume={6},
  number={1},
  pages={139--154},
  year={1978},
  publisher={Taylor \& Francis}
}

@article{lancichinetti2008benchmark,
  title={Benchmark graphs for testing community detection algorithms},
  author={Lancichinetti, Andrea and Fortunato, Santo and Radicchi, Filippo},
  journal={Physical review E},
  volume={78},
  number={4},
  pages={046110},
  year={2008},
  publisher={APS}
}

@article{bron1973algorithm,
  title={Algorithm 457: finding all cliques of an undirected graph},
  author={Bron, Coen and Kerbosch, Joep},
  journal={Communications of the ACM},
  volume={16},
  number={9},
  pages={575--577},
  year={1973},
  publisher={ACM New York, NY, USA}
}

@article{tomita2006worst,
  title={The worst-case time complexity for generating all maximal cliques and computational experiments},
  author={Tomita, Etsuji and Tanaka, Akira and Takahashi, Haruhisa},
  journal={Theoretical computer science},
  volume={363},
  number={1},
  pages={28--42},
  year={2006},
  publisher={Elsevier}
}

@article{cazals2008note,
  title={A note on the problem of reporting maximal cliques},
  author={Cazals, Fr{\'e}d{\'e}ric and Karande, Chinmay},
  journal={Theoretical computer science},
  volume={407},
  number={1-3},
  pages={564--568},
  year={2008},
  publisher={Elsevier}
}

@techreport{hagberg2008exploring,
  title={Exploring network structure, dynamics, and function using NetworkX},
  author={Hagberg, Aric and Swart, Pieter J and Schult, Daniel A},
  year={2008},
  institution={Los Alamos National Laboratory (LANL), Los Alamos, NM (United States)}
}

@inproceedings{konect,
	title = {{KONECT} -- {The} {Koblenz} {Network} {Collection}},
	author = {J\'{e}r\^{o}me Kunegis},
	year = {2013},
	booktitle = {Proc. Int. Conf. on World Wide Web Companion},
	pages = {1343--1350},
	url = {http://dl.acm.org/citation.cfm?id=2488173},
}

@inproceedings{pei2005mining,
  title={On mining cross-graph quasi-cliques},
  author={Pei, Jian and Jiang, Daxin and Zhang, Aidong},
  booktitle={Proceedings of the eleventh ACM SIGKDD international conference on Knowledge discovery in data mining},
  pages={228--238},
  year={2005}
}

@article{yu2023fast,
  title={Fast maximal quasi-clique enumeration: A pruning and branching co-design approach},
  author={Yu, Kaiqiang and Long, Cheng},
  journal={Proceedings of the ACM on Management of Data},
  volume={1},
  number={3},
  pages={1--26},
  year={2023},
  publisher={ACM New York, NY, USA}
}

@inproceedings{paranjape2017motifs,
  title={Motifs in temporal networks},
  author={Paranjape, Ashwin and Benson, Austin R and Leskovec, Jure},
  booktitle={Proceedings of the tenth ACM international conference on web search and data mining},
  pages={601--610},
  year={2017}
}

@inproceedings{kantorovich1942translocation,
  title={On the translocation of masses},
  author={Kantorovich, Leonid V},
  booktitle={Dokl. Akad. Nauk. USSR (NS)},
  volume={37},
  pages={199--201},
  year={1942}
}

@article{lin2002divergence,
  title={Divergence measures based on the Shannon entropy},
  author={Lin, Jianhua},
  journal={IEEE Transactions on Information theory},
  volume={37},
  number={1},
  pages={145--151},
  year={2002},
  publisher={IEEE}
}

@article{kullback1951information,
  title={On information and sufficiency},
  author={Kullback, Solomon and Leibler, Richard A},
  journal={The annals of mathematical statistics},
  volume={22},
  number={1},
  pages={79--86},
  year={1951},
  publisher={JSTOR}
}

@book{conover1999practical,
  title={Practical nonparametric statistics},
  author={Conover, William Jay},
  year={1999},
  publisher={john wiley \& sons}
}

@article{carraghan1990exact,
  title={An exact algorithm for the maximum clique problem},
  author={Carraghan, Randy and Pardalos, Panos M},
  journal={Operations Research Letters},
  volume={9},
  number={6},
  pages={375--382},
  year={1990},
  publisher={Elsevier}
}

@article{chang2024maximum,
  title={Maximum k-plex computation: Theory and practice},
  author={Chang, Lijun and Yao, Kai},
  journal={Proceedings of the ACM on Management of Data},
  volume={2},
  number={1},
  pages={1--26},
  year={2024},
  publisher={ACM New York, NY, USA}
}

@article{mahdavi2014branch,
  title={A branch-and-bound approach for maximum quasi-cliques},
  author={Mahdavi Pajouh, Foad and Miao, Zhuqi and Balasundaram, Balabhaskar},
  journal={Annals of Operations Research},
  volume={216},
  number={1},
  pages={145--161},
  year={2014},
  publisher={Springer}
}

@article{trukhanov2013algorithms,
  title={Algorithms for detecting optimal hereditary structures in graphs, with application to clique relaxations},
  author={Trukhanov, Svyatoslav and Balasubramaniam, Chitra and Balasundaram, Balabhaskar and Butenko, Sergiy},
  journal={Computational Optimization and Applications},
  volume={56},
  number={1},
  pages={113--130},
  year={2013},
  publisher={Springer}
}

@misc{kim2026appendix,
  author = {Song Kim and Hyewon Kim and Kaiqiang Yu and Taejoon Han and Junghoon Kim and Susik Yoon and Jungeun Kim},
  title = {{Online appendix of Efficient Computation of Maximum Flexi-Clique in Networks}},
  howpublished = {Available at: \url{https://github.com/Song1940/MaximumFlexiClique/blob/main/Appendix.pdf}}
}
